\documentclass[12pt]{article}
\usepackage[bookmarks,bookmarksnumbered]{hyperref}
\usepackage{amsmath,XoohmE}
\usepackage{graphicx,color}
\usepackage{booktabs,multirow}
\usepackage{url,algorithm,algorithmic}

\newcommand{\codeand}{\mathbin{\&}} 
\newcommand{\codexor}{\boxplus} 
\newcommand{\codeplus}{\codexor} 

\newcommand{\wt}{\mathop{\text{\rm wt}}\nolimits}

\newcommand{\Aut}{\mbox{Aut}}
\newcommand{\cA}{\mathcal{A}}

\newcommand{\bF}{\relax\leavevmode\hbox{$F$\kern-.55em
                  \vrule height1.95ex depth-1.8ex width5.75pt}\kern1pt}

\newcommand{\ddt}{\partial_\tau}

\newcommand{\sM}{{\mathscr{M}}}
\newcommand{\Spin}{\mathop{\text{Spin}}\nolimits}

\newcommand{\fT}{\text{\sf T}}

\newtheorem{theorem}{Theorem}[section]
\newtheorem{definition}{Definition}[section]
\newtheorem{proposition}[theorem]{Proposition}
\newtheorem{corollary}[theorem]{Corollary}
\long\def\oMit#1{}

\def\4{\hphantom{\hbox{$-$}}}
\def\8{\oplus}
\def\stk#1{{\tiny\shortstack{#1}}}
\def\vC#1{\vcenter{\hbox{\hss#1\hss}}}
\definecolor{Hey}{rgb}{1,0,.4}
\definecolor{orange}{rgb}{.8,.4,0}
\definecolor{plum}{rgb}{.4,0,.6}
\definecolor{bone}{rgb}{.95,.95,.9}
\definecolor{gray}{gray}{.25}

\def\Ic{\hbox{\sf c\kern-3pt\rule[.2pt]{.5pt}{4pt}\kern2.5pt}}
\catcode`@=11    
\newcounter{xmpl}\resetby{section}{xmpl}
\newenvironment{example}{\par\noindent\addtocounter{xmpl}{1}%
                          \def\@currentlabel{{\thesection.\arabic{xmpl}}}%
                          {\bfseries Example~\arabic{section}.\arabic{xmpl}}%
                           ~\ignorespaces\small}
                        {{\color{gray}\hrulefill\rule{.45pt}{1ex}}\par}
\def\Ft#1{\footnote{#1}}
\newdimen\parshift\parshift=\parindent
 \long\def\@footnotetext#1{\insert\footins{\reset@font\footnotesize\interlinepenalty%
  \interfootnotelinepenalty\splittopskip\footnotesep\splitmaxdepth\dp\strutbox%
   \floatingpenalty\@MM\hsize\columnwidth\addtolength{\hsize}{-2\parshift}
    \@parboxrestore\protected@edef\@currentlabel{\csname p@footnote\endcsname\@thefnmark}
      \color@begingroup
       \@makefntext{\rule\z@\footnotesep\ignorespaces#1\@finalstrut\strutbox\vglue1mm}
        \color@endgroup}}
 \long\def\@makefntext#1{\hglue\parshift
                         \vbox{\noindent\hb@xt@0em{\hss\@makefnmark\,}#1}\vglue2ex}
\catcode`@=12			
 \font\rOpe=cmsy10                        
 \def\ktl{{\hbox{\rOpe\char'170}}}        
 \def\kbl{{\hbox{\rOpe\char'170}}}        
 \def\kcr{{\reflectbox{\rOpe\char'170}}}        
 \def\ktr{{\reflectbox{\rOpe\char'170}}}        
 \def\kbr{{\reflectbox{\rOpe\char'170}}}        
 \def\Border{\vbox{\hsize0pt
        \setlength{\unitlength}{1mm}
        \newcount\xco
        \newcount\yco
        \xco=-21
        \yco=12
        \begin{picture}(0,0)(-7.5,0)
        \put(\xco,\yco){$\ktl$}
        \advance\yco by-1
        {\loop
        \put(\xco,\yco){$\kcr$}
        \advance\yco by-2
        \ifnum\yco>-240
        \repeat
        \put(\xco,\yco){$\kbl$}}
        \xco=170
        \yco=12
        \put(\xco,\yco){$\ktr$}
        \advance\yco by-1
        {\loop
        \put(\xco,\yco){$\kcr$}
        \advance\yco by-2
        \ifnum\yco>-240
        \repeat
        \put(\xco,\yco){$\kbr$}}
        \put(-19.5,13){\scalebox{.598}{State University of New York
            Physics Department|University of Maryland Center for
            String and Particle  Theory \&\ Physics Department|%
            Delaware State University DAMTP}}
        \put(-19.5,-241.5){\scalebox{.728}{University of Washington
            Mathematics Department|Pepperdine University Natural
            Sciences Division|Bard College Mathematics
            Department}}
        \end{picture}
        \par\vskip-8mm}}
\definecolor{UMred}{rgb}{.9,.05,.2}
 \def\UMbanner{\vbox{\hsize0pt
        \setlength{\unitlength}{.4mm}
        \thicklines
        \begin{picture}(0,0)(-30,-10)
        \put(165,16){\line(1,0){4}}
        \put(170,16){\line(1,0){4}}
        \put(180,16){\line(1,0){4}}
        \put(175,0){\line(1,0){4}}
        \put(180,0){\line(1,0){4}}
        \put(185,0){\line(1,0){4}}
        \put(169,0){\line(0,1){16}}
        \put(170,0){\line(0,1){16}}
        \put(179,0){\line(0,1){16}}
        \put(180,0){\line(0,1){16}}
        \put(184,0){\line(0,1){16}}
        \put(185,0){\line(0,1){16}}
        \put(169,16){\oval(8,32)[bl]}
        \put(170,16){\oval(8,32)[br]}
        \put(179,0){\oval(8,32)[tl]}
        \put(185,0){\oval(8,32)[tr]}
        \end{picture}
        \par\vskip-6.5mm
        \thicklines}}

 \SfTitles 
 \allowdisplaybreaks
 \seceq

\begin{document}
\thispagestyle{empty}
\vbox{\Border\UMbanner}
 \noindent
 \today \hfill{UMDEPP 08-010, SUNY-O/667}
  \vfill
 \begin{center}
{\LARGE\sf\bfseries\boldmath
  Topology Types of Adinkras and the Corresponding\\[3mm]
  Representations of $N$-Extended Supersymmetry}\\[3mm]
  \vfill
{\sf\bfseries C.F.\,Doran$^a$, M.G.\,Faux$^b$, S.J.\,Gates, Jr.$^c$,
     T.\,H\"{u}bsch$^d$,\\ K.M.\,Iga$^e$, G.D.\,Landweber$^f$, and R.L.\,Miller$^g$}\\[3mm]
{\small\it
  \leavevmode\hbox to\hsize{\hss
  $^a$Department of Mathematics,
      University of Washington, Seattle, WA 98105%
  , {\tt  doran@math.washington.edu}\hss}
  \\
  $^b$Department of Physics,
      State University of New York, Oneonta, NY 13825%
  , {\tt  fauxmg@oneonta.edu}
  \\
  $^c$Center for String and Particle Theory,\\[-1mm]
      Department of Physics, University of Maryland, College Park, MD 20472%
  , {\tt  gatess@wam.umd.edu}
  \\
  $^d$Department of Physics \&\ Astronomy,
      Howard University, Washington, DC 20059%
  , {\tt  thubsch@howard.edu}\\[-1mm]
  \leavevmode\hbox to\hsize{\hss
  DAMTP, Delaware State University, Dover, DE 19901%
  , {\tt  thubsch@desu.edu}\hss}
  \\
  $^e$Natural Science Division,
      Pepperdine University, Malibu, CA 90263%
  , {\tt  Kevin.Iga@pepperdine.edu}
  \\
 $^f$Mathematics Program, Bard College,
     Annandale-on-Hudson, NY 12504-5000%
  , {\tt  gregland@bard.edu}
  \\
 $^g$Department of Mathematics,\\[-1mm]
      University of Washington, Seattle, WA 98105%
  , {\tt  rlmill@math.washington.edu}
 }\\[3mm]
  \vfill
{\sf\bfseries ABSTRACT}\\[3mm]
\parbox{145mm}{We present further progress toward a complete classification scheme for describing supermultiplets of $N$-extended worldline supersymmetry, which relies on graph-theoretic topological invariants. In particular, we demonstrate a relationship between Adinkra diagrams and quotients of $N$-dimensional cubes, where the quotient groups are subgroups of $(\mathbb{Z}_2)^N$. We explain how these quotient groups correspond precisely to doubly even binary linear error-correcting codes, so that the classification of such codes provides a means for describing equivalence classes of Adinkras and therefore supermultiplets.  Using results from coding theory we exhibit the enumeration of these equivalence classes for all cases up to 26 supercharges, as well as the maximal codes, corresponding to minimal supermultiplets, for up to 32 supercharges.
 }
  
\end{center}
  \vfill\vfill

\clearpage
\setcounter{page}{0}
\pagenumbering{arabic}
\section{Introduction, Review and Synopsis}
 \label{IRS}
In the past decade there has been a renewed interest among physicists in certain basic questions in supersymmetry and, in particular, the classification of supersymmetric theories.  Dualities and advances in string theory and its M- and F-Theory extensions have raised an awareness of a range of supersymmetric theories that were not previously well-studied. In many cases, dimensional reduction has been a useful---and sometimes the only---way to analyze these theories. In particular, dimensional reduction to one dimension (time) is often particularly enlightening. This reduction may, in fact, retain all the information necessary to reconstruct its higher-dimensional preimage, by the reverse process of dimensional ``oxidization''\cite{rGR0}. Independently, $N$-extended worldline supersymmetry also governs the dynamics of wave functionals in supersymmetric quantum field theories and so applies to all of them also in this other, more fundamental way.

In addition, the past decade has seen a renewed interest by mathematicians in physical supersymmetry, in part due to the corresponding interest by physicists, and in part because of new induced advances in algebraic geometry\cite{rMMYau1,rMMYau2,rMMYau3,rMS} and four-dimensional topology\cite{rDon,rTQFT,rSW} (see Ref.\cite{rKISurv} for a survey accessible to physicists).  This has brought about a desire to approach the foundations of the subject more systematically.

\subsection{Adinkras}
Studies of one-dimensional supersymmetric theories have led to the development of the ${\cal GR}(d,N)$ algebras\cite{rGR1,rGR2,rGLP}, to the {\em\/Adinkras\/}\cite{rA}, and to other efforts\cite{rFKS,rFS1,rFS2,rKRT,rBKMO,rBG,rBKLS}. The main goal of these works is a comprehensive, constructive and {\em\/conveniently usable\/} classification of representations of supersymmetry. Herein, we focus on supersymmetry with no central charge, and examine those representations which may be described using Adinkras.

Adinkras are directed graphs with various colorings and other markings on vertices and edges, which in a pictorial way encode all details of the supersymmetry transformations on the component fields within a supermultiplet\cite{rA}. They can be related to superspace constructions\cite{r6-1,r6-2}, aid the understanding of complex systems of transformation rules\cite{r6-3a,r6-4}, and provide a more systematic and complete classification tool for representations of supersymmetry.

The study of Adinkras, then, naturally leads to the question of their classification.\footnote{It must be cautioned that there are two ways in which a classification of Adinkras falls short of classifying $D=1$ off-shell supermultiplets: first, not every such supermultiplet can be described by an Adinkra, and second, sometimes such a supermultiplet can be described by more than one Adinkra.  Nevertheless, the classification of Adinkras is an important step in such a program.}  This classification of Adinkras naturally falls into four steps:

\begin{enumerate}
\item Determine which topologies are possible (the topology of an Adinkra is the underlying graph of vertices and edges without colorings, as, for instance, in Ref.\cite{r6-1}).
\item Determine the ways in which vertices and edges may be colored.  The topology of the Adinkra, together with the colorings of vertices and edges, will be called the {\em chromotopology} of the Adinkra.  It is chromotopologies that are classified in this paper.
\item Determine the ways in which edges may be chosen as dashed or solid.  This is closely related to the well-known Clifford algebra theory, and will be studied in a future effort.
\item Determine the ways in which arrows may be directed along each edge.  This issue is addressed in Ref.\cite{r6-1}, and shown to be equivalent to the question of ``hanging'' the graph on a few sinks.  Alternately, we can start with an Adinkra where all arrows go from bosons to fermions, then perform a sequence of vertex raises to arrive at other choices of arrow directions.
\end{enumerate}

As it happens, it is convenient to do 1.~and 2.~together; that is, to classify chromotopologies.  Herein, we show that the classification of Adinkra chromotopologies is equivalent to another interesting question from coding theory: the classification of doubly even binary linear codes.  Much work has already been done in this area\cite{rMcWS,rCHVP,rCPS}, and the work described in this paper goes even further in developing this classification; see Appendix~\ref{app:RM}.

We emphasize that we focus here on the supersymmetric representation theory, not the dynamics of supersymmetric theories. This is logically necessary in any classification endeavor, as we need to first know the full palette of supersymmetric representations before discussing the properties of the dynamics in theories built upon such representations. Clearly, presupposing a standard, uncoupled Lagrangian for the supermultiplets that we intend to classify would necessarily limit the possibilities; there do exist supermultiplets which can only have interactive Lagrangians\cite{rHSS,rGSS}. Herein, we defer the task of finding Lagrangians involving the supermultiplets considered in this paper. In Refs.\cite{r6-2,r6-7a}, we have in fact started on such studies, and, using Adinkras, have constructed supersymmetric Lagrangians for some of the supermultiplets that are also discussed herein.

In units where $\hbar=1=c$, all physical quantities may have at most units of mass, the exponent of which is called the {\em\/engineering dimension\/} and is an essential element of physics analysis in general. The engineering dimension of a field $\f(\t)$ will be written $[\f]$; for more details, see Refs.\cite{r6-1,r6--1}.

\subsection{Main Result}
 \label{sMR}
Our main result about the chromotopology types of Adinkras and the corresponding supermultiplets, up to direct sums, may be summarized as follows:

We define the function:
\begin{equation}
 \vk(N):=
  \begin{cases}
   0 &\text{for $N<4$},\\
   1 &\text{for $N=4,5$}, \\
   2 &\text{for $N=6$}, \\
   3 &\text{for $N=7$}, \\
   4 + \vk(N{-}8) &\text{for $N\geq 8$, recursively}.
  \end{cases}
 \label{eKmax}
\end{equation}
\begin{enumerate}\itemsep=-3pt\vspace{-3mm}
 \item Every Adinkra can be separated into its connected components.  (The supermultiplet corresponding to such an Adinkra breaks up into a direct sum of other supermultiplets, each of which corresponds to one of the connected components of the Adinkra).
 \item Each such component has the topology of a $k$-fold ($0\leq k\leq\vk\6(N)$) iterated $\ZZ_2$-quotient of the $N$-cube.  (The corresponding supermultiplet is likewise a quotient of the supermultiplet corresponding to the $N$-cube).
 \begin{enumerate}\itemsep=-3pt
  \item The result of iterating these $\ZZ_2$-quotients is equivalent to quotienting by a subgroup of $(\ZZ_2)^N$, that is, by a binary linear $[N,k]$ code.
  \item A code can be used for such a quotient if and only if it is a doubly even code.
  \item There is a one-to-one correspondence between chromotopologies of Adinkras and doubly even codes, under which column-permutation equivalence in codes corresponds to swapping colors on the edges of the Adinkra, and $R$-symmetries of the supermultiplet.
  \item The number of distinct doubly even codes grows combinatorially with $N$ and $k$, and distinct codes produce distinct Adinkra chromotopologies.
  \item \label{l:Ineq}A large subset of distinct doubly even codes are also permutation-inequivalent; two permutation-inequivalent codes produce distinct Adinkra topologies. This occurs for $N\geq8$, but for $N\geq10$ this occurs even when $k=\vk(N)$, giving rise to topologically distinct {\em\/maximal quotients\/}, \ie, {\em\/minimal supermultiplets\/}. The number of these also grows combinatorially with $N,k$; see Table~\ref{t:G3}.
 \end{enumerate}
\end{enumerate}\vspace{-3mm}

It is possible to peruse these results backwards: One may start with the minimal supermultiplets in the entry~(\ref{l:Ineq}) above, and then reconstruct the non-minimal ones by reversing the quotienting procedure described herein. We will not explore the combinatorially complex mechanics of such procedures here, but it reflects the very large numbers of $[N,k]$-codes in Table~\ref{t:G3}.\ping

The paper is organized as follows:
 Section~\ref{s:C} is a brief introduction to codes, and Section~\ref{adinkras} provides a review of Adinkras and their relationship with supermultiplets.  The main material is in Section~\ref{AdiTop}, where Adinkra topologies are classified in terms of doubly even codes, which we explore further in Section~\ref{s:codes-redux}. 
  We then turn to discuss the subtleties in identifying Adinkras with supermultiplets, which lead to the second part of our main result presented in full detail in Section~\ref{CliffHanger}.

 A number of technical details are deferred to the Appendices.
 Details of the supersymmetry action in real supermultiplets are discussed in Appendix~\ref{app:SA}, while those of the complex case occupy Appendix~\ref{app:complex}.
 Finally, Appendix~\ref{app:RM} presents the necessary details of the classification of doubly even binary linear error-correcting codes.

\section{Codes}
 \label{s:C}
We begin with a brief introduction to the theory of codes so as to make it easier to spot the emergence of its elements in the subsequent sections. The Reader who wants a more thorough introduction to coding theory can consult a standard reference such as Refs.\cite{rMcWS,rCHVP,rCPS}.

Generally, a code is a set of strings of characters, each of which is given some kind of meaning in order to communicate between two people (say).  We do not assume these are secret codes, in that we are not necessarily trying to hide this communication from a third party.  Given a fixed alphabet, there are many possible strings of letters from this alphabet, but a code will specify only some of these as codewords.  In this paper we do not assign meaning to the codewords, and thus we consider a code simply to be a set of codewords.

In this paper, we will focus on {\em\/binary block codes}.  In this case, the available characters from which the codewords are composed are the binary digits $0$ and $1$, and we fix a positive integer $N$, and require that each codeword have length $N$.  Thus, a binary block code of length $N$ is a subset of $\{0,1\}^N$.  Though the standard notation for an element of a cartesian product is $(x_1,x_2,\cdots,x_N)$, in practice we frequently abandon the parentheses and the commas, so that the element $(0,1,1,0,1)$ may be written more succinctly as the codeword $01101$.  The components of such an $N$-tuple are called bits, and the $N$-tuple is called a binary number.

If we think of $\{0,1\}$ as the group $\ZZ_2$, and $\{0,1\}^N$ as the $N$-fold product of this group with itself, then a binary block code is called {\em\/linear} when it is a subgroup of $\{0,1\}^N$.  This is equivalent to the statement that if $v$ and $w$ are codewords, then $v\codexor w$ (their bitwise sum modulo $2$, or equivalently their exclusive or) is also a codeword.  For instance, if $N=3$, then $\{010,011\}$ is a binary block code of length 3, but it is not linear because the linear combination, $010\codexor 011=001$, is itself not one of the listed codewords.  On the other hand, $\{000,001,010,011\}$ is a linear code.

Now, $\ZZ_2$ is not only a group; it is also a field, so $\{0,1\}^N$ can be viewed as a vector space over $\{0,1\}$.  All the concepts of linear algebra then apply, but with $\mathbb{R}$ replaced with $\ZZ_2$.  Elements of $\{0,1\}^N$ may be thought of as vectors, with vector addition the operation $\codexor$ of bitwise addition modulo 2, scalar-multiplication by $0$ setting every bit to zero, and scalar-multiplication by $1$ leaving the vector unchanged.  As in standard linear algebra, there is the concept of a basis.  Given a linear subspace (that is, a linear code), we can find a finite set of codewords $g_1,\ldots, g_k$ so that every codeword can be written uniquely as a sum
\begin{equation}
\sum_{i=1}^k x_i g_i,
\end{equation}
where the coefficients $x_1,\ldots,x_k$ are each either $0$ or $1$.  The set $g_1,\ldots,g_k$ is then a basis, but in coding theory, it is also called a generating set for the linear code.  As in real linear algebra, there is no expectation that a generating set is uniquely determined by the linear code, but the number $k$, called the {\em\/dimension} of the code, is always the same for a given linear code. It is common to say we have an $[N,k]$ linear code when $N$ is the length of the codewords and $k$ is the dimension. It is traditional to denote a generating set as an $N{\times}k$ matrix, where each row is an element of the generating set.

If $v\in\{0,1\}^N$, we define the {\em\/weight} of $v$, written $\wt(v)$, to be the number of $1$s in $v$.  For instance, the weight of $01101$ is $\wt(01101)= 3$.

A binary linear block code is called {\em\/even} if every codeword in the code has even weight.  It is called {\em\/doubly even} if every codeword in the code has weight is divisible by $4$.  Examples of doubly even codes are given in Section~\ref{s:deC} below.

If $v$ and $w$ are in $\{0,1\}^N$, then $v\codeand w$ is defined to be the ``bitwise and'' of $v$ and $w$: the $i$th bit of $v\codeand w$ is $1$ if and only if the $i$th bit of $v$ and the $i$th bit of $w$ are both $1$.  A basic fact in $\{0,1\}^N$ is
\begin{equation}
\wt(v \codexor w)=\wt(v)+\wt(w)-2\,\wt(v\codeand w).\label{eqn:inclusionexclusion}
\end{equation}
There is a standard inner product.  If we write $v$ and $w$ in $\{0,1\}^N$ as $(v_1,\ldots,v_N)$ and $(w_1,\ldots,w_N)$, then
\begin{equation}
\langle v,w\rangle \equiv \sum_{i=1}^N v_i\,w_i\pmod{2}.\label{eqn:codeinnerproduct}
\end{equation}
We call $v$ and $w$ orthogonal if $\langle v,w\rangle=0$.  This occurs whenever there are an even number of bit positions where both $v$ and $w$ are $1$.  Note that $\langle v,v\rangle\equiv\wt(v)\pmod{2}$, and thus, when $\wt(v)$ is even, $v$ is orthogonal to itself.  Also note that $\langle v,w\rangle \equiv \wt(v\codeand w)\pmod{2}$.  One important consequence for us is that if $\wt(v)$ and $\wt(w)$ are multiples of 4, then  (\ref{eqn:inclusionexclusion}) implies that $\wt(v \codexor w)$ is a multiple of 4 if and only if $v$ and $w$ are orthogonal.

\section{Supersymmetric Representations and Adinkras}
 \label{adinkras}
The $N$-extended supersymmetry algebra without central charges in one dimension is generated by the time-derivative, $\ddt$, and the $N$ supersymmetry generators, $Q_1, \ldots, Q_N$, satisfying the following supersymmetry relations:
\begin{equation}
 \big\{\,Q_I\,,\,Q_J\,\big\}=2\,i\,\d_{IJ}\,\ddt,\quad
  \big[\,\ddt\,,\,Q_I\,\big] =0,\quad  I,J=1,\ldots,N. \label{eSuSy}
\end{equation}
In this section we determine some essential facts about the transformation rules of these operators on fields for which it is possible to maintain the physically motivated concept of engineering dimension.
We note that since the time-derivative has engineering dimension $[\ddt]=1$, the supersymmetry relations\eq{eSuSy} imply that the engineering dimension of the supersymmetry generators is $[Q_I]=\inv2$.

\subsection{Supermultiplets as Representations of Supersymmetry}
A real supermultiplet $\sM$ is a real, unitary, finite-dimensional, linear representation of the algebra\eq{eSuSy}, in the following sense: It is spanned by a basis of real bosonic and fermionic {\em\/component fields\/}, $\f_1(\t),\ldots,\f_m(\t)$ and $\j_1(\t),\ldots,\j_m(\t)$, respectively; each component field is a function of time, $\t$. The supersymmetry transformations, generated by the Hermitian operators $Q_1,\cdots,Q_N$, act linearly on $\sM$ while satisfying Eqs.\eq{eSuSy}, \ie, Eqs.\eq{e[dQ,dQ]}, below.
 The supermultiplet is {\em\/off-shell\/} if no differential equation is imposed on it\Ft{Logically, it is possible for some---but not all---component fields to become subject to a differential equation. This does not violate the literal definition of the off-shell supermultiplet. However, it does obstruct standard methods of quantization, which is our eventual purpose for keeping supermultiplets off-shell. For an example in 4-dimensional supersymmetry, see Ref.\cite{r6-4}.}.
 The number of bosons as fermions is then the same, guaranteed by supersymmetry.

The full supersymmetry transformation, which preserves the reality of the component fields $\f_A(\t)$ and $\j_A(\t)$ and is generated by the $Q_I$, takes the form:
\begin{equation} 
 \begin{bmatrix} \f_A(\t)\\ \j_A(\t) \end{bmatrix} ~\mapsto~ e^{\d_Q (\e)}\> 
 \begin{bmatrix} \f_A(\t)\\ \j_A(\t) \end{bmatrix},\quad 
\d_Q(\e) := -i\e^I\,Q_I,\quad 
\f_A(\t),\j_A(\t) \in \sM, 
\label{edQ(e)} 
\end{equation} 
where $\e^I$ are the (anticommuting) parameters of the transformation. The infinitesimal transformation operators, $\d_Q(\e)$, in turn satisfy the relation:
\begin{equation}
 \big[\,\d_Q(\e)\,,\,\d_Q(\h)\,\big] = 2\,i\,(\e^I\d_{IJ}\,\h^J)\,\ddt.
 \label{e[dQ,dQ]}
\end{equation}
Both the relations\eq{eSuSy} and\eq{e[dQ,dQ]} equivalently specify the supersymmetry algebra, and the infinitesimal transformation operator $\d_Q(\e)$ may be---and, in the physics literature, is---used to study the supersymmetry transformations in supermultiplets.

\subsection{Building Supermultiplets from Adinkras}
 \label{sCloSuSy}
Refs.\cite{rA,r6-1,r6-2} introduced and then studied Adinkras, diagrams that encode the transformation rules\eq{edQ(e)} of the component fields under the action of the supersymmetry generators $Q_1,\ldots,Q_N$.

\Remk\label{r:1}
Classical Lie groups have one, or in the case of $\Spin(2n)$ two, {\em\/fundamental\/} irreducible representations from which all (infinitely many) other finite-dimensional, unitary representations can be constructed by tensoring the fundamental one(s), symmetrizing in all possible ways and subtracting ``traces'' using the invariant tensors of the group. The analogous method using Adinkras was sketched in Refs.\cite{r6-1,r6-3c}. Thus, Adinkras and the corresponding supermultiplets should represent at least the {\em\/fundamental\/} representations of supersymmetry, from which to build the infinite tower of all others, very much akin to the case of classical Lie algebras.

Without further ado, we focus on the supermultiplets that {\em\/can\/} be depicted by Adinkras:

\begin{defn}\label{dAd}
A supermultiplet $\sM$ is {\em\bfseries\/adinkraic\/} if it admits a basis, $(\f_1,\cdots,\f_m\,|\,\j_1,\cdots,\j_m)$,  of component fields such that each $Q_I\in\{Q_1,\cdots,Q_N\}$ acts upon each $\f_A\in\{\f_1,\cdots,\f_m\}$ so as to produce:\vspace{-2mm}
\begin{subequations}\label{eQbf}
 \begin{align}
 Q_I \, \f_A(\t) &= c\,\ddt^{\l}\, \j_B(\t),\quad\text{where}\quad
  c=\pm1,\quad \l\in\{0,1\},\quad\j_B\in\{\j_1,\cdots,\j_m\},
   \label{eQb}\\[-2mm]
\intertext{and the right-hand side choices depend on $I$ and $A$. In turn, this $Q_I$ acting on this $\j_B$ produces:\vspace{-2mm}}
 Q_I \, \j_B(\t) &= \frac{i}{c}\,\ddt^{1-\l}\, \f_A(\t),\label{eQf}
\end{align}
\end{subequations}
and the pair of formulae\eq{eQbf} exhausts the action of each $Q_I$ upon each component field.
\end{defn}

Every adinkraic supermultiplet may be depicted as an Adinkra:

\begin{defn}\label{dA}
The {\em\bfseries\/Adinkra\/}, $\cA_\sM$, of an adinkraic supermultiplet $\sM$ is a labeled, directed graph, $(W,E,C,O,D)$, where $W$ is a bipartitioned set of vertices ($\,V$): one half white, the other black; $E$ is the set of edges, each edge connecting two vertices of opposite color; $C$ the set of colorings of the edges; $O$ the set of their orientations; and $D$ the set of their dashedness.

 A vertex is assigned to each component field of $\sM$, white for bosons, black for fermions.
 A transformation rule of the form\eq{eQb} corresponds to an $I$-colored edge connecting the vertex corresponding to $\f_A$ to the vertex corresponding to $\j_B$. If $c=-1$, the edge is dashed, and the edge is oriented from $\f_A$ to $\j_B$ if $\l=0$ and the other way around if $\l=1$; see Table~\ref{t:A}.
\end{defn}
\begin{table}[ht]
  \centering
  \begin{tabular}{@{} cc|cc @{}}
    \makebox[15mm]{\bf Adinkra} & \makebox[40mm]{\bf\boldmath$Q$-action} 
  & \makebox[15mm]{\bf Adinkra} & \makebox[40mm]{\bf\boldmath$Q$-action} \\ 
    \hline
    \begin{picture}(5,9)(0,5)
     \put(0,0){\includegraphics[height=11mm]{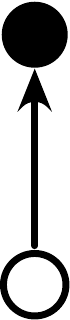}}
     \put(3,0){\scriptsize$A$}
     \put(3,9){\scriptsize$B$}
     \put(-1,4){\scriptsize$I$}
    \end{picture}\vrule depth4mm width0mm
     & $Q_I\begin{bmatrix}\j_B\\\f_A\end{bmatrix}
           =\begin{bmatrix}i\dot\f_A\\\j_B\end{bmatrix}$
  & \begin{picture}(5,9)(0,5)
     \put(0,0){\includegraphics[height=11mm]{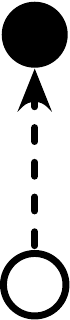}}
     \put(3,0){\scriptsize$A$}
     \put(3,9){\scriptsize$B$}
     \put(-1,4){\scriptsize$I$}
    \end{picture}\vrule depth4mm width0mm
     & $Q_I\begin{bmatrix}\j_B\\\f_A\end{bmatrix}
           =\begin{bmatrix}-i\dot\f_A\\-\j_B\end{bmatrix}$ \\[5mm]
    \hline
    \begin{picture}(5,9)(0,5)
     \put(0,0){\includegraphics[height=11mm]{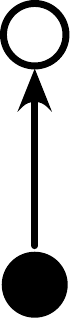}}
     \put(3,0){\scriptsize$B$}
     \put(3,9){\scriptsize$A$}
     \put(-1,4){\scriptsize$I$}
    \end{picture}\vrule depth4mm width0mm
     &  $Q_I\begin{bmatrix}\f_A\\\j_B\end{bmatrix}
           =\begin{bmatrix}\dot\j_B\\i\f_A\end{bmatrix}$
  & \begin{picture}(5,9)(0,5)
     \put(0,0){\includegraphics[height=11mm]{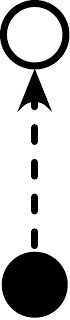}}
     \put(3,0){\scriptsize$B$}
     \put(3,9){\scriptsize$A$}
     \put(-1,4){\scriptsize$I$}
    \end{picture}\vrule depth4mm width0mm
     &  $Q_I\begin{bmatrix}\f_A\\\j_B\end{bmatrix}
           =\begin{bmatrix}-\dot\j_B\\-i\f_A\end{bmatrix}$ \\[5mm]
    \hline
  \multicolumn{4}{l}{\vrule height3.5ex width0pt\parbox{120mm}{\small The edges are here labeled by the variable index $I$; for fixed $I$, they are drawn in the $I^{\text{th}}$ color.}}
  \end{tabular}
  \caption{The correspondences between the Adinkra components and supersymmetry transformation formulae:
    vertices\,$\iff$\,component fields;
    vertex color\,$\iff$\,fermion/boson;
    edge color/index\,$\iff$\,$Q_I$;
    edge dashed\,$\iff$\,$c=-1$; and
    orientation\,$\iff$\,placement of $\ddt$.
    They apply to all $\f_A,\j_B$ within a supermultiplet and all $Q_I$-transformations amongst them.}
  \label{t:A}
\end{table}

We can also use the Adinkra to reconstruct the adinkraic supermultiplet, since the Adinkra contains all the information necessary to write down the transformation rules. Thus, an Adinkra is simply a graphical depiction of the transformation rules\eq{eQbf}.

\Remk\label{r:2}
Owing to Remark~\ref{r:1}, we do not expect the class of {\em\/adinkraic supermultiplets\/} to exhaust the representations of supersymmetry\eq{eSuSy}, but rather to serve as the {\em\/fundamental\/} representations from which to build all the other ones. In addition to the usual construction (tensoring, symmetrizing and taking traces), this will also involve the question when is it is possible for {\em\/all\/} the supersymmetry transformations in a supermultiplet to conform to Eqs.\eq{eQbf} {\em\/simultaneously\/}; we will address this under separate cover.
 In addition, certain exceptional adinkraic supermultiplets may be depicted, for high enough $N$, by more than one Adinkra. A case-by-case determination if two given Adinkras correspond to the same supermultiplet is not too difficult, but a systematic study is deferred to a subsequent effort\cite{r6-3.2}.

In view of Remarks~\ref{r:1}--\ref{r:2}, the fact that most important physical examples {\em\/can\/} be described using Adinkras reinforces their interpretation as {\em\/fundamental\/} and provides the underlying motivation for our present study. We therefore focus on classifying those supermultiplets that do correspond to Adinkras, and so the question of which supermultiplets cannot be so depicted remains outside our present scope. In turn, Ref.\cite{r6-1} shows that all Adinkras sharing the same topology and their corresponding supermultiplets---and indeed their superfield representations---may be obtained from any one of them; it then remained to classify the topologies available to Adinkras, up to Remarks~\ref{r:1}--\ref{r:2}, and we now turn to that.

\section{Adinkra Chromotopologies}
 \label{AdiTop}
It is the purpose of this paper to classify the possible {\em\/topologies\/} for Adinkras. To this end, we will need the precise definition:

\begin{definition}\label{dA2T}
The {\em\bfseries\/topology\/} of an Adinkra, $\fT(\cA_\sM)$, is the graph $(V,E)$ consisting of only the (unlabeled) vertices and edges of the Adinkra; cf.\ Definition~\ref{dAd}. Also, $\fT(\sM):=\fT(\cA_\sM)$.
\end{definition}
In particular, from Definition~\ref{dA} of an Adinkra, we forget the bipartition (black or white) of the vertices and all additional information associated to the edges, namely, the edge-coloring (to which $Q_I$ corresponds), dashedness ($c=\pm1$), and direction (the exponent of $\ddt$ in Eqs.{\rm\eq{eQbf}}, and so in fact all $\ddt$).
\begin{definition}\label{dA2C}
The {\em\bfseries\/chromotopology\/} of an Adinkra is the topology of the Adinkra, together with the vertex bipartition (coloring each vertex black or white), and the edge coloring (assigning a color to each edge).
\end{definition}
In particular, from Definition~\ref{dA} of an Adinkra,  we forget the dashedness of the edges and the direction of the arrow along the edge.

\subsection{Cubical Adinkras}
The fundamental example of an Adinkra topology is that of the $N$-cube, $I^N=[0,1]^N$. It has $2^N$ vertices and $N{\cdot}2^{N-1}$ edges. We may embed it in $\IR^N$ by locating the vertices at the points $\vec{p}=(p_1,\cdots,p_N)\in\IR^N$, where $p_I=0,1$ in all $2^N$ possible combinations.  An edge connects two vertices that differ in precisely one coordinate. For every vertex, $\vec{p}$, the {\em\/weight of $\vec{p}$}, written $\wt(\vec{p})$, equals the number of $J\in\{1,\cdots,N\}$ for which $p_J=1$.

There is a natural chromotopology associated to the $N$-cube: 
 As the weights of the vertices are either odd or even, we color them either black or white, respectively; flipping this choice is called the `Klein flip'.
 We associate the numbers from $1$ to $N$ with $N$ different colors, and then color each edge of the $N$-cube that connects vertices which differ only in the $I$th coordinate with the $I$th color.  The result is called the colored $N$-cube.

We now construct four supermultiplets: $\sM^\diamond_{I^N}$, $\sM^=_{I^N}$ and their Klein-flips, all with the chromotopology of a colored $N$-cube, $I^N$.  

\subsubsection{The Exterior Supermultiplet}
 \label{s:EM}
We start with a bosonic, real field $\f_0(\t)$, write
\begin{equation}
 Q^{\vec{p}}:=Q_1^{p_1}\cdots Q_N^{p_N}, \label{eMQ}
\end{equation}
and define all other component fields in reference to $\f_0(\t)$:
\begin{equation}
 \sM^\diamond_{I^N}:=
 \big\{\, F_{\vec{p}}(\t)
       :=(-i)^{\wt(\vec{p})\choose2}Q^{\vec{p}}\,\f_0(\t),~
 \vec{p}\in I^N \,\big\}.
 \label{eM_IN}
\end{equation}
Since
\begin{equation}
 \left((-i)^{\wt(\vec{p})\choose2}\,Q^{\vec{p}}\right)^\dagger
 =\left((+i)^{\wt(\vec{p})\choose2}\right)
  \left((-1)^{\wt(\vec{p})\choose2}Q^{\vec{p}}\right)
 =(-i)^{\wt(\vec{p})\choose2}\,Q^{\vec{p}}, \label{eRF}
\end{equation}
the component fields $F_{\vec{p}}(\t)$ in \Eq{eM_IN} are real.
 Each vertex $\vec{p}$ is marked by a white node if $\wt(\vec{p})$ is even, and \Eq{eM_IN} assigns it the bosonic component field $F_{\vec{p}}(\t)$; when $\wt(\vec{p})$ is odd, the assigned $F_{\vec{p}}(\t)$ is fermionic and its node is black. This binary ``cubist'' notation, $F_{\vec{p}}$, is not hard to translate into the one more familiar to physicists: each component, $p_I=0,1$, specifies if the index $I$ is absent or present, respectively, in the antisymmetrized multi-index of the component field. Thus, {\em\/e.g.,\/} $F_{000}=F=\f_0(\t)$, $F_{100}=F_1$, $F_{101}=F_{[13]}$, $F_{011}=F_{[23]}$, $F_{111}=F_{[123]}$, where the antisymmetrization on multiple indices stems from the anticommutivity of the $Q_I$'s.
 
Next, using that
\begin{equation}
 Q_I\,Q^{\vec{p}}=(-1)^{\wt(\vec{p}<I)}\,(i\ddt)^{\vec{p}\cdot\vec{e}_I}\,
  Q^{\vec{p}\,\codexor\,\vec{e}_I}, \label{eQIQp}
\end{equation}
where $\wt(\vec{p}<I)$ counts the number of nonzero components of $\vec{p}$ before the $I^{\text{th}}$, $\vec{e}_I$ is the unit vector in the $I^{\text{th}}$ direction and\begin{equation}
 \vec{p}\,\codexor\vec{e}_I\equiv(\vec{p}+\vec{e}_I)\pmod2, \label{ep+eI}
\end{equation}
is the component-wise exclusive or, we compute by direct application of $Q_I$:\begin{equation}
 Q_I\,F_{\vec{p}}(\t)=
  (-1)^{\wt(\vec{p}<I)+(\vec{p}\cdot\vec{e}_I)(\wt(\vec{p})+1)}\,
  i^{\wt(\vec{p})}\,
  \big(\ddt^{(\vec{p}\cdot\vec{e}_I)}F_{\vec{p}\,\codexor\vec{e}_I}(\t)\big).
 \label{eQFp}
\end{equation}
Thus, for each $I$, the $Q_I$-transformation connects precisely those two vertices that differ in the $I^{\text{th}}$ coordinate. The corresponding edge in the Adinkra, drawn in the $I^{\text{th}}$ color, is oriented from the vertex where that coordinate is 0 to where it is 1. It is either solid or dashed depending on the exponents of $(-1)$ and $i$.
 This makes the topology of the $N$-cube manifest and also proves that $\sM^\diamond_{I^N}$ is adinkraic, seeing that \Eq{eQFp} conforms precisely to Eqs.\eq{eQbf}.

The Klein-flip of $\sM^\diamond_{I^N}$ is obtained by starting from a fermion, $\j_0(\t)$, in place of a boson, $\f_0(\t)$.

 The engineering dimensions of the component fields of $\sM^\diamond_{I^N}$ are:
\begin{equation}
 [F_{\vec{p}}\,]=w_0+\inv2\wt(\vec{p}), \label{eEDT}
\end{equation}
where $w_0$ is a constant determined by the choice of a Lagrangian. 
 The number of component fields, listed by their increasing engineering dimension, determines the $\ZZ$-graded dimension of this representation:
\begin{equation}
 \begin{aligned}
  \dim(\sM^\diamond_{I^N})
  &=\dim\big(F_{0\cdots0}=\f_0|
             F_{10\cdots0},F_{010\cdots0},\cdots|
             \cdots|F_{11\cdots1}\big),\\
  &=\Big({\ttt 1\big|N\big|{N\choose2}\big|{N\choose3}\big|\cdots\big|
         {N\choose N-1}\big|{N\choose N}}\Big),
  \end{aligned}
 \label{eESMd}
\end{equation}
so that there are $2^{N-1}$ bosonic and $2^{N-1}$ fermionic component fields in $\sM^\diamond_{I^N}$.

Adinkras corresponding to $\sM^\diamond_{I^N}$ were called Top Clifford Algebra superfields\cite{rA}. The supermultiplets themselves will be recognized by supersymmetry practitioners to be also representable as real, ``unconstrained'' Salam--Strathdee superfields, in the familiar $\q$-expansion of which $F_{\vec{p}}(\t)$ occurs as the coefficient of the $\q_1^{p_1}\cdots \q_N^{p_N}$ monomial. Since the superspace $\q$'s generate an {\em\/exterior\/} algebra, we refer to $\sM^\diamond_{I^N}$ as the {\em\/exterior supermultiplet\/}.  

In closing this section, let us note the definition of the component fields introduced above is sufficient for the purposes of the present work as our goals simply involved studying the representation theory of one-dimensional supersymmetry.  If we were to consider the case of some dynamical theory, \ie, write a superfield Lagrangian to specify some dynamics, then an alternative definition of the components would be necessary, as was done in Ref.\cite{r6-2}, for example.  In all manifestly supersymmetrical theories, there exist `twisted'  versions of the supercharges denoted by ${\rm D}_1, \ldots, {\rm D}_N$, satisfying the following relations:
\begin{equation}
 \big\{\,{Q}_I\,,\,{\rm D}_J\,\big\}=0,\quad  
 \big\{\,{\rm D}_I\,,\,{\rm D}_J\,\big\}=2i\,\d_{IJ}\,\ddt,\quad
  \big[\,\ddt\,,\,{\rm D}_I\,\big] =0,\quad  I,J=1,\ldots,N. \label{eSuSytw}
\end{equation}
An operator analogous to (4.1) may be defined by
\begin{equation}
 {\rm D}^{\vec{p}} \Defl {\rm D}_1^{p_1}\cdots {\rm D}_N^{p_N}, \label{eMQ1}
\end{equation}
and this is applied to superfields to define components upon taking the limit as all Grassmann coordinate
vanish.  This entire process is known as defining components by projection.

\subsubsection{The Clifford Supermultiplet}
 \label{s:CM}
Another important supermultiplet, $\sM^=_{I^N}:=\{\bF_{\vec{p}}(\t),\vec{p}\in\{0,1\}^N\}$, may be obtained from $\sM^\diamond_{I^N}$ {\em\/via\/} the {\em\/non-local\/} transform:
\begin{equation}
  \bF_{\vec{p}}(\t):=
   \ddt^{\lfloor(N-\wt(\vec{p}))/2\rfloor}
    \,F_{\vec{p}}(\t),\qquad\text{\ie},\qquad
  F_{\vec{p}}(\t)=\ddt^{\lfloor(\wt(\vec{p})-N)/2\rfloor}\,\bF_{\vec{p}}(\t).
 \label{eTop2Base}
\end{equation}
\Remk\label{r:NL}
Since the transformation $\sM^\diamond_{I^N} \iff \sM^=_{I^N}$ is non-local for $N>1$, these two supermultiplets must be regarded as distinct. With increasing $N$, it is clear that there is also a combinatorially growing multitude of ``intermediate'' supermultiplets, all having the same, $N$-cubical topology, and obtainable one from another by means of non-local transformations of the kind\eq{eTop2Base}\cite{r6-1}.

Notice that, in contrast to \Eq{eEDT},
\begin{equation}
 [\bF_{\vec{p}}\,]=
 \begin{cases}
  w_0+\lfloor\frac{N}2\rfloor &\text{if $\wt(\vec{p})$ is even},\\
  w_0+\lfloor\frac{N}2\rfloor+\inv2 &\text{if $\wt(\vec{p})$ is odd}.
 \end{cases}
 \label{eEDB}
\end{equation}
The Adinkra corresponding to the supermultiplet $\sM^=_{I^N}$ was called the ``base'' Adinkra in Ref.\cite{rA}.
Also, all off-shell bosonic component fields have the same engineering dimension, $w_0$, and so do all fermionic ones, equal to $w_0{+}\inv2$.
 The number of component fields, listed by their increasing engineering dimension, is $(2^{N-1}|2^{N-1})$---considerably simpler than \Eq{eESMd}.

\begin{example}
The simplest nontrivial example of a Clifford supermultiplet occurs at $N=2$ and consists of two bosons $\f_1,\f_2$ and two fermions $\j_1,\j_2$, where the two supersymmetry generators act as follows:
\begin{subequations}\label{N2IS}
 \begin{align}
 Q_1\,\f_1&=\j_1,      & Q_2\,\f_1&=\j_2,\label{Qf1}\\*
 Q_1\,\f_2&=\j_2,      & Q_2\,\f_2&=-\j_1,\label{Qf2}\\*
 Q_1\,\j_1&=i\dot\f_1, & Q_2\,\j_1&=-i\dot\f_2,\label{Qj1}\\*
 Q_1\,\j_2&=i\dot\f_2,
  \begin{picture}(30,0)(-12,0)
  \put(10,0){\includegraphics[height=20mm]{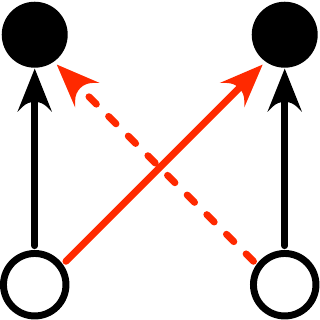}}
  \put(6,0){\small$\f_1$}
  \put(6,20){\small$\j_1$}
  \put(30,0){\small$\f_2$}
  \put(30,20){\small$\j_2$}
 \end{picture}\quad
                        & Q_2\,\j_2&=i\dot\f_1,\label{Qj2}
 \end{align}
\end{subequations}
The black (vertical) edges represent $Q_1$-action and red (diagonal) edges that of $Q_2$.
Forgetting either one of $Q_1,Q_2$, we remain with two {\em\/disjoined\/} $N=1$ supermultiplets. The $Q_2$-action involves here a scaling factor of $c=-1$ in one pair, depicted by a dashed edge. This could be remedied by redefining $\j_1$, for instance, but only to induce a $c=-1$ scaling factor in the $Q_1$-action $\f_1\iff\j_1$.  No such redefinition will eliminate the need for a minus sign in at least one pair of the transformations\eqs{Qf1}{Qj2}.

The $N=2$ exterior supermultiplet and Adinkra is:
\begin{subequations}\label{N2ES}
 \begin{align}
 Q_1\,\f&=\j_1,      & Q_2\,\f&=\j_2,\label{Q.f1}\\*
 Q_1\,\j_1&=i\dot\f, & Q_2\,\j_1&=-iF,\label{Q.j1}\\*
 Q_1\,\j_2&=iF, & Q_2\,\j_2&=i\dot\f,\label{Q.j2}\\*
 Q_1\,F&=\dot\j_2,
  \begin{picture}(30,0)(-12,0)
  \put(10,-2){\includegraphics[height=25mm]{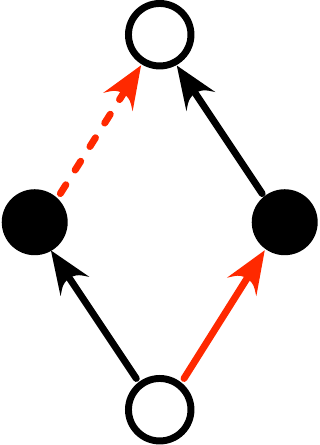}}
  \put(13,-1){\small$\f$}
  \put(5,10){\small$\j_1$}
  \put(29,10){\small$\j_2$}
  \put(22,20){\small$F$}
 \end{picture}\quad
                        & Q_2\,F&=-\dot\j_1,\label{Q.f2}
 \end{align}
\end{subequations}
Notice that $(\f,F|\j_1,\j_2)\to(\f_1,\dot\f_2|\j_1,\j_2)$ and $(\f_1,\f_2|\j_1,\j_2)\to(\f,\ddt^{-1}F|\j_1,\j_2)$ is the manifestly non-local bijective field redefinition between the two supermultiplets.

 Together with their Klein-flips, these are all the $N=2$ Adinkras and supermultiplets\cite{rA,r6-1}.
\end{example}

\begin{example}
As another example, let $N=3$. For $\sM^\diamond_{I^3}$, denote the component fields $F_{\vec{p}}$ alternatively as $\f_{\vec{p}}$ for $\wt(\vec{p})$ even (bosons), and  $\j_{\vec{p}}$ for $\wt(\vec{p})$ odd (fermions). For $\sM^=_{I^3}$, denote the component fields $\bar\f_{\vec{p}}$ and $\bar\j_{\vec{p}}$, respectively, where \Eq{eTop2Base} implies, for example:
\begin{equation}
 \bar\f_{000}=\dot\f_{000},\qquad
 \bar\j_{100}=\dot\j_{100},\qquad
 \bar\f_{110}=\f_{110},\qquad
 \bar\j_{111}=\j_{111}.
\end{equation}
Following Definition~\ref{dA}, we then draw:
\begin{equation}
 \vC{
 \begin{picture}(50,40)(0,-5)
 \put(-20,18){\large$\sM^\diamond_{I^3}$~:}
 \put(-3,-3){\includegraphics[height=35mm]{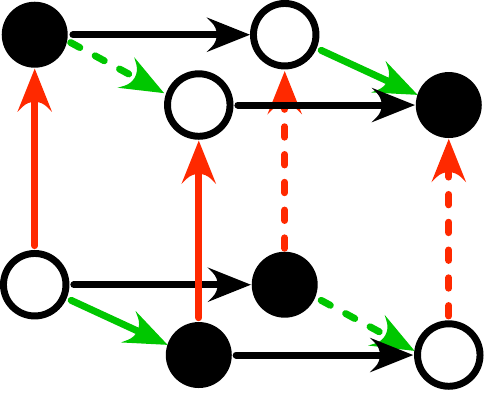}}
 \put(-3,30){\llap{$\j_{010}$}}
 \put(-3,8){\llap{$\f_0=\f_{000}$}}
 \put(12,-2){\llap{$\j_{001}$}}
 \put(3,20){$\f_{011}$}
 \put(25,9){$\j_{100}$}
 \put(27,30){$\f_{110}$}
 \put(41,-1){$\f_{101}$}
 \put(41,22){$\j_{111}$}
 \end{picture}}
 \quad\text{\it vs.}~~\qquad
  \vC{
 \begin{picture}(50,40)(-15,-5)
 \put(-20,18){\large$\sM^=_{I^3}$~:}
 \put(-3,-3){\includegraphics[height=35mm]{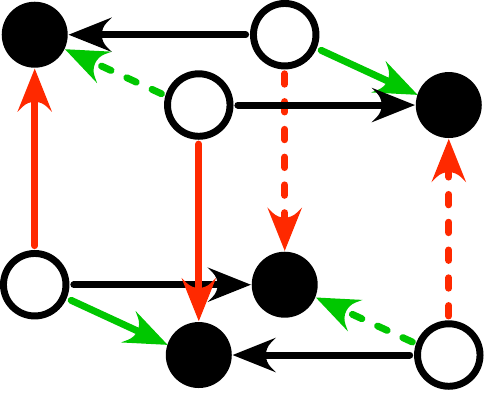}}
 \put(-3,30){\llap{$\bar\j_{010}$}}
 \put(-3,8){\llap{$\dot\f_0=\bar\f_{000}$}}
 \put(12,-2){\llap{$\bar\j_{001}$}}
 \put(3,20){$\bar\f_{011}$}
 \put(25,9){$\bar\j_{100}$}
 \put(27,30){$\bar\f_{110}$}
 \put(41,-1){$\bar\f_{101}$}
 \put(41,22){$\bar\j_{111}$}
 \end{picture}}
 \label{eCube}
\end{equation}
The distinction between these becomes clearer in the convention of Ref.\cite{r6-1}, with the nodes drawn at a height proportional to the engineering dimension of the corresponding component field:
\begin{equation}
  \vC{
  \begin{picture}(50,60)(0,0)
  \put(-17,30){\large$\sM^\diamond_{I^3}$~:}
  \put(0,2){\includegraphics[height=55mm]{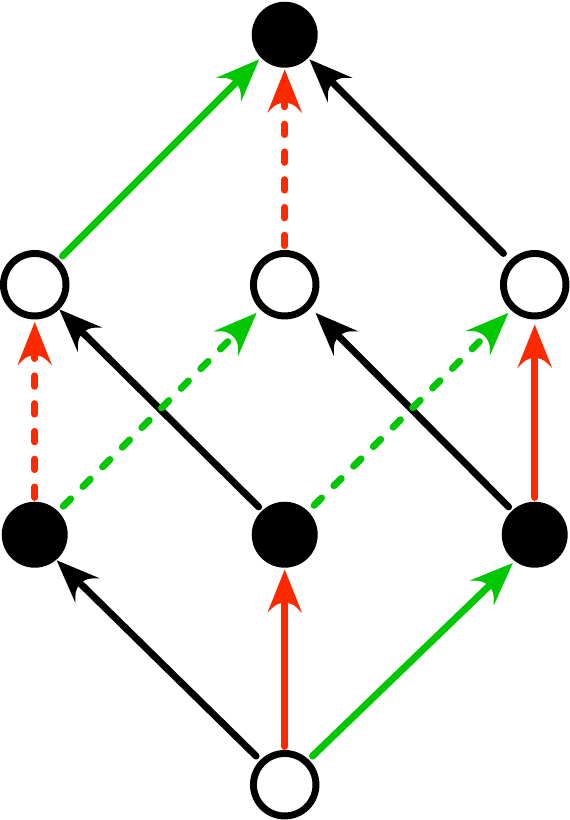}}
   \put(-3,3){$\f_0=\f_{000}$}
   \put(-10,22){$\j_{100}$}
   \put(7,22){$\j_{010}$}
   \put(38,22){$\j_{001}$}
   \put(-10,40){$\f_{110}$}
   \put(7,40){$\f_{101}$}
   \put(38,40){$\f_{011}$}
   \put(7,55){$\j_{111}$}
 \end{picture}} \qquad\text{\it vs.}\qquad\qquad
  \vC{
  \begin{picture}(50,60)(0,-10)
  \put(-13,10){\large$\sM^=_{I^3}$~:}
  \put(0,2){\includegraphics[height=20mm]{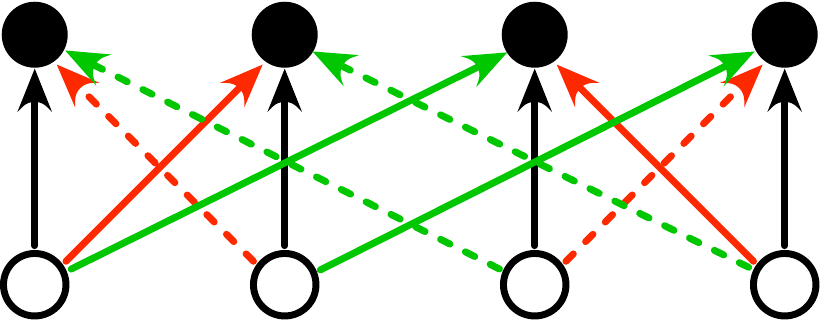}}
   \put(-3,-1){$\bar\f_{000}$}
   \put(13,-1){$\bar\f_{110}$}
   \put(29,-1){$\bar\f_{101}$}
   \put(45,-1){$\bar\f_{021}$}
   \put(-3,24){$\bar\j_{100}$}
   \put(13,24){$\bar\j_{010}$}
   \put(29,24){$\bar\j_{001}$}
   \put(45,24){$\bar\j_{111}$}
 \end{picture}}
 \label{eCubED}
\end{equation}
The Adinkra corresponding to the Klein-flip of $\sM^\diamond_{I^N}$ is obtained from \Eq{eCube} by reversing the vertex coloring, black$\,\iff\,$white.
\end{example}

We now turn to determine what other topologies are available to Adinkras, and correspondingly to supermultiplets.

\subsection{Quotients of the $N$-Cubes}
 \label{cubeproject}
Given an $N$-cubical Clifford supermultiplet, $\sM^=_{I^N}$, a novel opportunity emerges for $N\geq4$: Certain identifications amongst the $2^{N-1}$ bosons and the $2^{N-1}$ fermions are possible such that the induced action of the $Q_1,\cdots,Q_N,\ddt$ upon the components of the so projected supermultiplet remains consistent with Eqs.\eq{eSuSy}.
 This is {\em\/not\/} possible for $N<4$.
\begin{example}
For example, attempting to identify $\f_1=\f_2$ in Eqs.\eq{N2IS} leads either to a trivial supermultiplet consisting of a single bosonic constant, or to an immediate contradiction between the left-hand side and the right-hand side equations\eq{Qf1} and\eq{Qf2}:
\begin{align}
 Q_1\,\f_1(\t)&=\j_1(\t), & Q_2\,\f_1(\t)&=\j_2(\t),\tag{\ref{Qf1}$'$}\\
 Q_1\,\f_1(\t)&=\j_2(\t), & Q_2\,\f_1(\t)&=-\j_1(\t).\tag{\ref{Qf2}$'$}
\end{align}
From the left-hand side pair, it follows that $\j_2(\t)=+\j_1(\t)$, whereas the right-hand side pair implies $\j_2(\t)=-\j_1(\t)$. This is consistent only if $\j_1(\t)=0=\j_2(\t)$, whereupon \Eq{eSuSy} implies that $\ddt\f_1(\t)=0$, reducing this supermultiplet to a trivial, single constant boson.  Alternately, adding the left-hand side Eq.~(\ref{Qf1}$'$) to the right-hand side Eq.~(\ref{Qf2}$'$) leads to $ Q_1\,\f_1 +  Q_2\,\f_1$ = 0.  In a similar manner, subtracting the  second result in Eq.~(\ref{Qf1}$'$) from the first result in Eq.~(\ref{Qf2}$'$) leads to $ Q_1\,\f_1 - Q_2\,\f_1 = 0$.
 Together, these imply that $ Q_1\,\f_1 =  Q_2\,\f_1 = 0$. 
\end{example}

\subsubsection{The $N=4$ Projection}
To clarify how such projections {\em\/can\/} occur, we examine the simplest of them, for $N=4$.
\begin{example}
\label{PrHyCube}
The $N=4$ Isoscalar supermultiplet (\ie, the Clifford supermultiplet with bosons on the bottom) is
\begin{equation}
  \sM^=_{I^4}=\{\,\bar\f_{0000},\bar\f_{1100},\ldots,\bar\f_{1111}|
                   \bar\j_{1000},\ldots,\bar\f_{1110},\ldots\,\},
 \label{eM_IN=}
\end{equation}
where the ellipses indicate component fields the subscript of which is obtained as a permutation on the one preceding the ellipses.
 The Adinkra of this supermultiplet is:
\begin{equation}
 \vC{
 \begin{picture}(120,55)(3,-5)
 \put(-15,42){\large$\sM^=_{I^4}$~:}
 \put(-3,-3){\includegraphics[height=50mm]{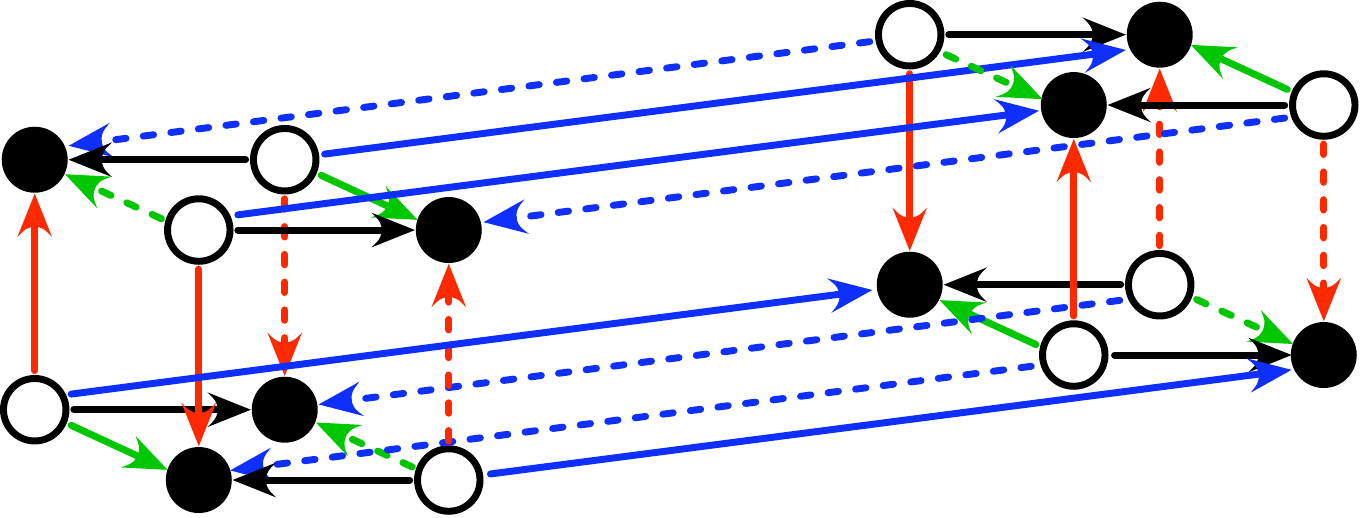}}
 \put(-14,32){$\bar\j_{0100}$}
 \put(-14,8){$\bar\f_{0000}$}
 \put(3,-2){$\bar\j_{0010}$}
 \put(3,21){$\bar\f_{0110}$}
 \put(26,11){$\bar\j_{1000}$}
 \put(29,30.5){$\bar\f_{1100}$}
 \put(44,-2){$\bar\f_{1010}$}
 \put(43,21){$\bar\j_{1110}$}
 \put(71,45){$\bar\f_{0101}$}
 \put(71,22){$\bar\j_{0001}$}
 \put(86,13){$\bar\f_{0011}$}
 \put(86,34){$\bar\j_{0111}$}
 \put(114,21){$\bar\f_{1001}$}
 \put(114,45){$\bar\j_{1101}$}
 \put(130,12){$\bar\j_{1011}$}
 \put(130,36){$\bar\f_{1111}$}
 \end{picture}}
 \label{e4Cube}
\end{equation}
By replacing, for convenience, in the bottom square of the right-hand side 3-cube:
\begin{alignat}{3}
 \bar\f_{0011}&\to-\bar\f_{0011}~, \quad&\quad
 \bar\f_{1001}&\to-\bar\f_{1001}~,\\
 \bar\j_{0001}&\to-\bar\j_{0001}~, \quad&\quad
 \bar\j_{1011}&\to-\bar\j_{1011}~,
\end{alignat}
the Adinkra\eq{e4Cube} becomes:
\begin{equation}
 \vC{
 \begin{picture}(120,55)(3,-5)
 \put(-15,42){\large$\p:\sM^=_{I^4}\to\sM^=_{I^4}$~:}
 \put(-3,-3){\includegraphics[height=50mm]{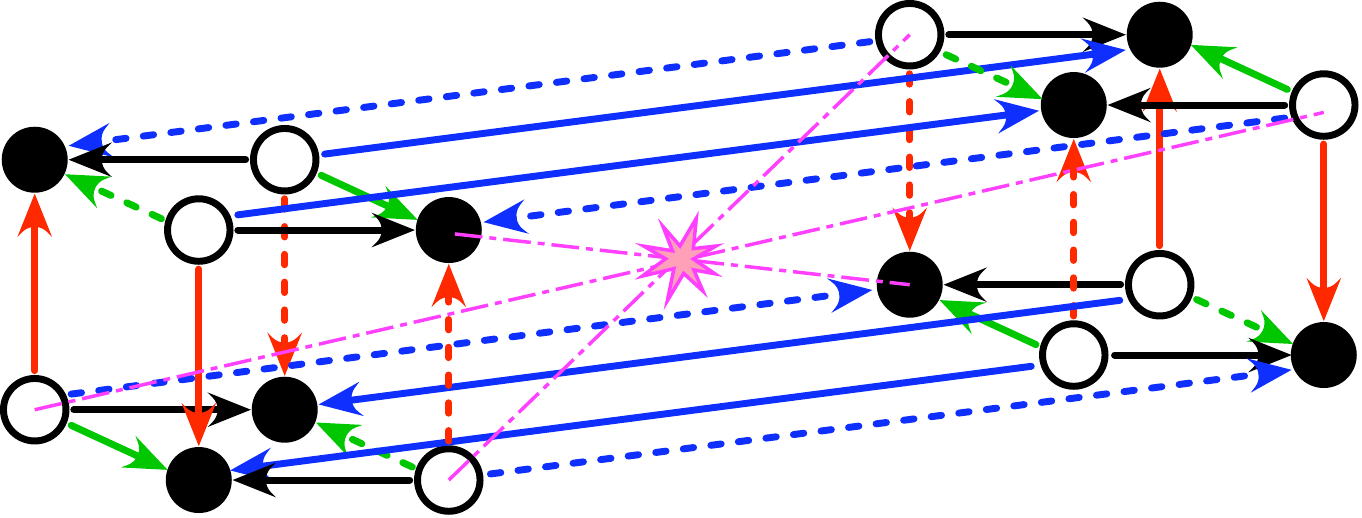}}
 \put(-14,32){$\bar\j_{0100}$}
 \put(-14,8){$\bar\f_{0000}$}
 \put(3,-2){$\bar\j_{0010}$}
 \put(3,21){$\bar\f_{0110}$}
 \put(26,11){$\bar\j_{1000}$}
 \put(29,30.5){$\bar\f_{1100}$}
 \put(44,-2){$\bar\f_{1010}$}
 \put(43,21){$\bar\j_{1110}$}
 \put(71,45){$\bar\f_{0101}$}
 \put(69,22){$-\bar\j_{0001}$}
 \put(84,13){$-\bar\f_{0011}$}
 \put(86,34){$\bar\j_{0111}$}
 \put(111,22){$-\bar\f_{1001}$}
 \put(114,45){$\bar\j_{1101}$}
 \put(129,12){$-\bar\j_{1011}$}
 \put(130,36){$\bar\f_{1111}$}
 \end{picture}}
 \label{e4CubeM}
\end{equation}
where the reflection across the central point marked by the star is a symmetry of the Adinkra, including all the labelings.  For the sake of clarity in this diagram, we have only indicated three of the eight vertex pairs being reflected into each other by the dot-dash purple lines. The full set of reflection pairs is:
\begin{subequations}
\begin{alignat}{7}
 \bar\f_{0000}  &\iff \bar\f_{1111}~, \quad&\quad
 -\bar\f_{0011}  &\iff \bar\f_{1100}~, \quad&\quad
 -\bar\j_{0001} &\iff \bar\j_{1110}~, \quad&\quad
 \bar\j_{0010} &\iff \bar\j_{1101}~,\\
 \bar\f_{0101}  &\iff \bar\f_{1010}~, \quad&\quad
 \bar\f_{1001}  &\iff \bar\f_{0110}~, \quad&\quad
 \bar\j_{0100} &\iff -\bar\j_{1011}~, \quad&\quad
 \bar\j_{1000} &\iff \bar\j_{0111}~.
\end{alignat}
\end{subequations}
It is thus possible to ``orbifold'' the Adinkra\eq{e4CubeM} by identifying, {\em\/e.g.\/}, $\bar\f^\pm_{0000}$ with $\bar\f_{0000}\pm\bar\f_{1111}$, and consequently also identifying each two such vertices and so also the corresponding component fields:
\begin{equation}
  \p~:\quad (p_1,p_2,p_3,p_4) ~\mapsto~ (1-p_1,1-p_2,1-p_3,1-p_4),
  \qquad p_I\in\{0,1\}. \label{eD4Proj}
\end{equation}
Both resulting Adinkras have the topology of (the 1-skeleton of) the projective hypercube in $N=4$:
\begin{equation}
 \vC{
 \begin{picture}(120,40)(-45,0)
 \put(-45,20){$(\sM^=_{I^4}/\ZZ_2)^\pm$~:}
 \put(-2,0){\includegraphics[height=35mm]{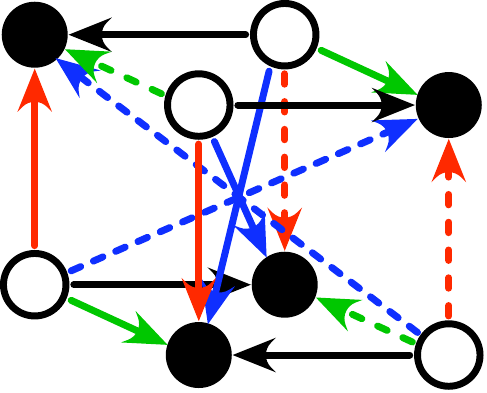}}
 \put(-13,33){$\bar\j^\pm_{0100}$}
 \put(-13,12){$\bar\f^\pm_{0000}$}
 \put(2,2){$\bar\j^\pm_{0010}$}
 \put(3,21){$\bar\f^\pm_{0110}$}
 \put(26,13){$\bar\j^\pm_{1000}$}
 \put(27,33){$\bar\f^\pm_{1100}$}
 \put(42,2){$\bar\f^\pm_{1010}$}
 \put(41,22){$-\bar\j^\pm_{0001}$}
 \end{picture}}
 \label{eRP4}
\end{equation}
where we have defined
\begin{alignat}{3}
 (\sM^=_{I^4}/\ZZ_2)^\pm&: &\quad
  \bF^\pm_{\vec{p}}(\t)&\Defl ~\bF_{\vec{p}}(\t)\pm(-1)^{p_2+p_4}
                           \bF_{\vec1-\vec{p}}(\t)
 \label{ePrCF}
\end{alignat}
for the component fields corresponding to the orbifolded Adinkras.
 As it turns out, the dimensional reduction of the $d=4$, ${\cal N}=1$ chiral supermultiplet has the topology\eq{eRP4}\cite{r6-1} but the height assignments (engineering dimensions) of two of its bosons differ from the ones shown here.
 Note that every boson is connected to every fermion, and {\em\/vice versa\/}.  This becomes clearer in the convention of Ref.\cite{r6-1}, where we place the vertices at a height proportional to the engineering dimensions of the corresponding component fields, so that all arrows point upward:
\begin{equation}
 \vC{
 \begin{picture}(120,45)(-30,-5)
 \put(-25,15){$(\sM^=_{I^4}/\ZZ_2)^\pm$~:}
 \put(2,2){\includegraphics[height=30mm]{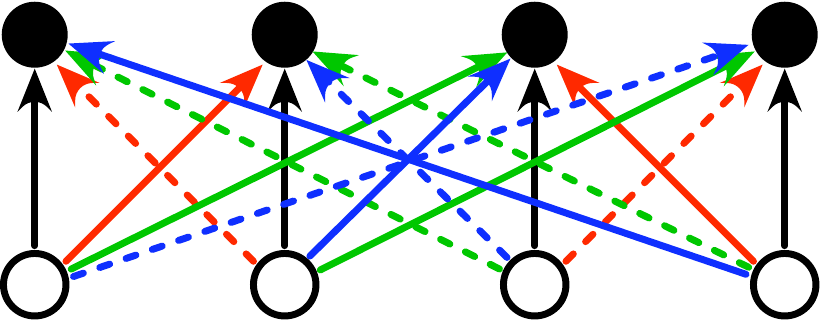}}
 \put(2,35){$\bar\j^\pm_{1000}$}
 \put(26,35){$\bar\j^\pm_{0100}$}
 \put(50,35){$\bar\j^\pm_{0010}$}
 \put(74,35){$-\bar\j^\pm_{0001}$}
 \put(-2,-1){$\bar\f^\pm_{0000}$}
 \put(22,-1){$\bar\f^\pm_{1100}$}
 \put(46,-1){$\bar\f^\pm_{1010}$}
 \put(70,-1){$\bar\f^\pm_{0110}$}
 \end{picture}}
 \label{eB44}
\end{equation}
 In graph theory, this is known as $K(4,4)$, the complete bipartite graph connecting four bosons with four fermions.  Though it may appear at first that there could be further identifications, a closer look at the Adinkra\eq{eB44} reveals that the locations of the dashed lines prevent all further identification; that is, no further projection is consistent with \Eq{eSuSy}.
\end{example}
\Remk\label{r:3}
 A combination of \Eq{ePrCF} and \Eq{eTop2Base} shows that attempting an analogous projection in $\sM^\diamond_{I^4}$ would require the definitions:
\begin{alignat}{3}
 \text{``$(\sM^\diamond_{I^4}/\ZZ_2)^\pm$''}&: &\quad
 F^\pm_{\vec{p}}(\t)
 &:=\ddt^{\lfloor(4-\wt(\vec{p}))/2\rfloor}F_{\vec{p}}(\t)
     \pm(-1)^{p_2+p_4}\ddt^{\lfloor(\wt(\vec{p}))/2\rfloor}F_{\vec1-\vec{p}}(\t).
 \label{ePrCFd}
\end{alignat}
The appearance of the $\ddt$'s indicates that this is not a quotient of the exterior supermultiplet $\sM^\diamond_{I^4}$ itself, but of a multiple {\em\/vertex-raise\/}\cite{r6-1} of $\sM^\diamond_{I^4}$, which is in turn isomorphic to the $\sM^=_{I^4}$. It should be manifest that the differing engineering dimensions of the component fields in $\sM^\diamond_{I^4}$ (and the locality requirements) present the key obstruction to such a $\ZZ_2$ projection in $\sM^\diamond_{I^4}$. This proves that the possibility to construct a local $\ZZ_2$-quotient of a supermultiplet strongly depends on its component fields' engineering degrees: while $\sM^=_{I^4}$ does have a local $\ZZ_2$-quotient, $\sM^\diamond_{I^4}$ does not.

\subsubsection{$N>4$ Projections}
More generally, given an Adinkra, we can sometimes identify vertices, if all labeling of the corresponding vertices and edges matches up. This results is a {\em\/quotient Adinkra\/}. A correspondingly projected supermultiplet is a {\em\/quotient supermultiplet\/}.  We will now focus on quotients of the $N$-cube.  We show in Section~\ref{sQuot} that all connected Adinkra topologies can be obtained in this way.

We start with an Adinkra with the topology of an $N$-cube, the vertices of which are the $N$-tuples $(x_1,\cdots,x_N)\in\{0,1\}^N$.  Suppose that the desired projection, $\p$, is to identify $\vec{0}=(0,\cdots,0)$ with some other vector $\vec{x}=(x_1,\cdots,x_N)$.  Now consider $(1,0,\cdots,0)$, which is connected to $\vec{0}$ by an edge of color $1$.  There is also an edge of the same color, $1$, incident with the vertex at $(x_1,\cdots,x_N)$, and this edge connects $(x_1,\cdots,x_N)$ to $(1-x_1,x_2,\cdots,x_N)$.  If the edges are to match up, $\p$ must also identify $(1,0,\cdots,0)$ with $(1-x_1,x_2,\cdots,x_N)$. Following similarly the edges of all other colors, we can prove by induction that if $(v_1,\cdots,v_N)$ is any vertex of $\{0,1\}^N$, then $\p$ identifies $(v_1,\cdots,v_N)$ with $(x_1+v_1,\cdots,x_N+v_N)\pmod{2}$; we write $\vec{v}\id_\pi\vec{v}\codexor\vec{x}$.  Thus, every projection $\p$ is completely determined by the vertex which it identifies with $\vec{0}$.

More generally, the projection $\p$ may identify more than one vertex with $\vec{0}$.  Suppose $\vec{x}$ and $\vec{y}$ are two such vertices, identified with $\vec{0}$.  Then $\p$ establishes an equivalence relation so that $\vec{0}\id_\p\vec{x}$, and $\vec{x}\id_\p\vec{x}\codexor\vec{y}$.  Thus, $\vec{0}$ is identified with $\vec{x}\codexor\vec{y}$.

Now, we note that the inverse $-\vec{x}$ is equal to $\vec{x}$ modulo 2.  Therefore, the set of vertices that are to be identified with $\vec{0}$ forms a group under component-wise addition modulo 2.  Since we must identify only bosons with bosons and only fermions with fermions, $\wt(\vec{x})$ must be even.  As we will see in \SS~\ref{sQuot}, the weight $\wt(\vec{x})$ must actually be a multiple of 4, or the dashedness of the edges cannot possibly match up.

This yields the following construction:
\begin{enumerate}\itemsep=-3pt\vspace{-3mm}
 \item Start with an Adinkra, $\cA$, with the topology an $N$-cube, $I^N=[0,1]^N$.
 \item Let $G$ be a subgroup of $(\ZZ_2)^N$ consisting only of vectors $\vec{x}$ with $\wt(\vec{x})\equiv 0\pmod4$.
 \item Let $G$ act on the vertex set $\{0,1\}^N$ by component-wise addition modulo 2.
 \item If this preserves the dashing of the edges, then $\cA/G$ is an Adinkra whose topology is $I^N/G$, the corresponding quotient of the $N$-cube.
\end{enumerate}\vspace{-3mm}

\subsection{Finding all Adinkra Chromotopologies}
 \label{sQuot}
Suppose an Adinkra is not connected.  Then the corresponding supermultiplet splits into a direct sum, each component of which corresponds to a connected component of the Adinkra.  Now it is possible for a connected Adinkra to correspond to a direct sum, though it may not be immediately apparent by the Adinkra. This issue will be discussed fully in Ref.\cite{r6-3.2}, where we also specify precisely the conditions under which a supermultiplet may be represented by two topologically distinct Adinkras.  For now, all we need to note is that the problem of classifying Adinkras reduces to classifying connected Adinkras.

As we will see, every connected Adinkra chromotopology is obtained by taking a colored $N$-dimensional cube, then possibly identifying nodes and edges using a doubly even code.  Thus, the topology of an Adinkra, obtained by forgetting the coloring of the edges and vertices, arises from quotienting an (uncolored) $N$-dimensional cube using a doubly even code.  First, some notation.

Suppose we have an Adinkra for $N$-extended supersymmetry, and we consider the chromotopology.  That is, we ignore the arrows on the edges and ignore whether an edge is dashed.  What is left is a vertex set $V=\{v_1,\cdots,v_{2m}\}$, corresponding to all component fields $(F_1\6(\t),\cdots,F_{2m}\6(\t))$,\footnote{Previously, we labeled the bosons $\f_1,\cdots,\f_m$ and the fermions $\j_1,\cdots,\j_m$. Here, it will be convenient for notation to treat them on the same footing.\label{introf}}, the coloring of these vertices, and the edge set $E$ with its coloring.  Since this graph is inherited from a proper Adinkra, for every vertex and every $I\in\{1,\dots,N\}$ there is an edge corresponding to applying $Q_I$ to the field corresponding to that vertex.  Applying this $Q_I$ may involve derivatives, a sign, and/or a factor of $i$, but it will be convenient for now to suppress this.  To this end, we define functions $q_1,\cdots,q_N$ from the vertex set $V=\{v_1,\cdots,v_{2m}\}$ to itself, such that whenever $A$, $B$, and $I$ are such that there is an equation of type (\ref{eQb}) or (\ref{eQf}),
\begin{equation}
  Q_I \, F_A(\t)= c\,\ddt^{\lambda} F_B,\qquad\implies\qquad
  q_I(v_A)=v_B.
\label{e:qAct}
\end{equation}
Notice that in defining $q_I$ from $Q_I$, we are forgetting the coefficients $c=\pm1$, as well as the $\ddt$'s which encode the differences in engineering dimensions.
 The supersymmetry algebra (\ref{eSuSy}) then implies:
\begin{alignat}{3}
  q_I^2&=\Ione,\qquad&\text{\ie}\qquad
  q_I\big(q_I(v)\big)&=v,~\text{for all $v\in V$},\label{e:q2=1}\\[-8mm]
\intertext{and\vspace{-4mm}}
  q_Iq_J&=q_Jq_I,\qquad&\text{\ie}\qquad
  q_I\big(q_J(v)\big)&=q_J\big(q_I(v)\big),~\text{for all $I,J$, and all $v\in V$}.
   \label{e:qiqj}
\end{alignat}

\begin{theorem}
Every connected Adinkra chromotopology is a quotient of a colored $N$-dimensional cube.
\label{T:QuoC}
\end{theorem}
\begin{proof}
Suppose we have a connected Adinkra with $N$ edge colors.  Let $V=\{v_1,\cdots,v_{2m}\}$ be the set consisting of all its vertices. Now pick, without loss of generality, any one bosonic vertex $v_*\in V$ and fix it.
 We take the colored $N$-cube $[0,1]^N$, and consider its vertex set $\{0,1\}^N$. We then define a mapping
\begin{equation}
 \p:\left\{ \begin{aligned}
               \{0,1\}^N        &\to V,\\
               (x_1,\cdots,x_N) &\mapsto q_1^{x_1}\big(\cdots q_N^{x_N}(v_*)\big),
             \end{aligned}\right.
 \label{ePrCEff}
\end{equation}
and write $\p(\vec{x}):=q_1^{x_1}\big(\cdots q_N^{x_N}(v_*)\big)$. 

For any $\vec{x}\in\{0,1\}^N$, we apply \Eq{ePrCEff} to $v_*$. Then, for any $I\in\{1,\cdots,N\}$, apply $q_I$ to $\p(\vec{x})$:
\begin{align}
 q_I\big(\p(\vec{x})\big)
 &=q_I\,\Big(q_1^{x_1}\big(\cdots q_N^{x_N}(v_*)\big)\Big)
  =q_1^{x_1}\Big(\cdots\,q_I^{1-x_I}\big(\cdots\,q_N^{x_N}(v_*)\big)\Big),\label{eQIout}\\[2mm]
 &= \p(\vec{x}\,\codexor\,\vec{e}_I),\qquad
  \forall v_*\in V,~\forall\vec{x}\in\{0,1\}^N,\label{eQIin}
\end{align}
where we have applied \Eq{e:qiqj} to commute $q_I$ through the other $q$'s in \Eq{eQIout} and so obtain \Eq{eQIin}, using the notation\eq{ep+eI} for the component-wise ``exclusive or'' operation.

Thus, each pair of vertices ($\p(\vec{x}),\p(\vec{x}\codexor\vec{e}_I))\in V\times V$ is connected by an edge, labeled by $I$, in the Adinkra.  Since there is an edge labeled by $I$ connecting $\vec{x}$ with $(\vec{x}\codexor\vec{e}_I)$ in the $N$-cube, and since all edges in the cube $[0,1]^N$ are of this form, the map\eq{ePrCEff} induces a map $\p_E$ that maps edges colored $I$ of $[0,1]^N$ to edges colored $I$ in the Adinkra.

To see that $\p$ is surjective, first observe that the Adinkra is connected.  That is, every vertex $v\in V$ is connected, via a path of edges, to the fixed $v_*\in V$. In the Adinkra, these edges have colors, forming a sequence, $I_1,\cdots,I_k$, when tracing from the vertex $v$ to $v_*$. If we then apply to $v_*$ a corresponding sequence of $q_I$'s, we get:
\begin{equation}
  q_{I_1}\big(\cdots q_{I_k}(v_*)\big)=v.
\end{equation}
Using the commutativity of the $q$'s to put them in numerical order and \Eq{e:q2=1} to eliminate those that appear more than once, we can write this as
\begin{equation}
 \p(\vec{x}) = q_1^{x_1}\big(\cdots q_N^{x_N}(v_*)\big) = v.
\end{equation}

The fact that $\p$ sends bosons to bosons and fermions to fermions can be seen by the fact that it sends the boson $(0,\dots,0)\in\{0,1\}^N$ to the boson $v_*$, and the fact that every vertex in $V$ is connected to $v_*$ by a sequence of edges, each of which alternates between bosons and fermions.

Let $G=\p^{-1}(v_*)$, be the collection of points $\vec{x}\in\{0,1\}^N$, which \Eq{ePrCEff} maps to $v_*$.  Being a subset of $\{0,1\}^N$, $G$ may be interpreted as a subset of $(\ZZ_2)^N$.
 
We will now show $G$ is a subgroup of $(\ZZ_2)^N$.  In fact, as we will see, $G$ is the group such that $[0,1]^N/G$ is the chromotopology of the Adinkra in question.  Trivially, $\vec0\in G$.  Inverses exist since every element is its own inverse in $(\ZZ_2)^N$.  Finally, let $\vec{x}$ and $\vec{y}$ be elements of $G$.  Compose $q_1^{x_1}\dots q_N^{x_N}$ with $q_1^{y_1}\dots q_N^{y_N}$.  The composition sends $v_*$ to itself, and we can commute the $q_I$ past each other to write $q_1^{x_1+y_1}\cdots q_N^{x_N+y_N}$.  Thus,
\begin{equation}
q_1^{x_1+y_1}\cdots q_N^{x_N+y_N}(v_*)=v_*.
\end{equation}
Since $q_I^2=1$, we can reduce the exponents of these modulo 2, and thus, $\pi(\vec{x}\codexor\vec{y})=v_*$.  Thus, $\vec{x}\codexor\vec{y}\in G$.  Therefore, $G$ is a group.

We now define a bijection between $\{0,1\}^N/G$ and the vertex set $V$ of the Adinkra.  The function $\pi:\{0,1\}^N\to V$ has the property that for all $\vec{x}\in G$, $\pi(\vec{y}+\vec{x})=\pi(\vec{y})$.  Thus it is well-defined to define a function on the cosets of $G$, $f:\{0,1\}^N/G\to V$, so that $f(\vec{y}+G)=\pi(\vec{y})$.

The function $f$ is one-to-one, since if $\pi(\vec{y})=\pi(\vec{z})$, we would have
\begin{equation}
q_1^{y_1}\dots q_N^{y_N}(v_*)=q_1^{z_1}\dots q_N^{z_N}(v_*),
\end{equation}
and this would imply
\begin{equation}
q_1^{y_1+z_1}\dots q_N^{y_N+z_N}(v_*)=v_*,
\end{equation}
so that $\vec{y}\codexor \vec{z}\in G$.
The function $f$ is onto, since $\p$ is, as was shown above.  Thus, $f$ is a bijection on vertices.  The argument that $\p$ sends edges to edges shows that the edges are similarly in bijection.
We therefore see that $G$ is equal to the group of identifications on $\{0,1\}^N$, and the orbit space $\{0,1\}^N/G$ is the Adinkra.  The map $\pi$ realizes this quotient.
We can use $\pi$ to pull back all the labels from the Adinkra to $\{0,1\}^N$, and in this way, we have a cubical Adinkra the quotient of which by $G$ is the Adinkra in question.
\end{proof}

We next consider what kinds of groups $G$ can be used to quotient a cubical Adinkra.  First, note that $G$ is a subgroup of $(\ZZ_2)^N$.  As we saw in Section~\ref{s:C}, $G$ is then a binary linear code. In the analysis to far, we have used only the edge-coloring in the Adinkra.

Remembering also the vertex-coloring will force $G$ to be even, that is, for each $\vec{x}\in G$, $\wt(\vec{x})$ must be even.  As we will see shortly, remembering the edge-dashedness will force $G$ to be {\em\/doubly\/} even, that is, for each $\vec{x}\in G$, $\wt(\vec{x})$ must be a multiple of 4.

\begin{theorem}
If an Adinkra is a quotient of an $N$-cube by a group $G$, then $G$ is a doubly even binary linear code.
\end{theorem}
\Remk
The converse is also true, that is if $G$ is a doubly even binary linear code, then there is a family of adinkraic supermultiplets, the chromotopology of the Adinkra of which is a quotient by $G$ of a colored $N$-cube. In order to prove this fact, we must first discuss the relationship between supersymmetry and Clifford algebras, which will be treated in the sequel, Ref.\cite{r6-3.2}, in which for every doubly even code, $G$, we construct a family of adinkraic supermultiplets with the $\{0,1\}^N/G$ chromotopology.

\begin{proof}
Recall from Section~\ref{s:C} that a binary linear code is a subgroup of $(\ZZ_2)^N$, and that it is called doubly even if every element of it has weight a multiple of 4.  So we need to prove that for every $x\in G$, $\wt(x)\equiv 0 \pmod{4}$.

To this end, the $q_I$'s no longer suffice, and we will need to use the $Q_I$'s; in particular, we need to ``remember'' the scaling constants $c$ encoding the edge-dashedness in an Adinkra, and the equipartition of the vertices into bosonic and fermionic ones. 

The statement that $\vec{x}\in G$ means that
\begin{equation}
 \p(\vec{x})=q_1^{x_1}\big(\cdots q_N^{x_N}(v_*)\big)=v_*,
\end{equation}
and thus
\begin{equation}
 Q_1^{x_1}\cdots Q_N^{x_N}\, F_1(\t) = C\,\ddt^{\,\wt(\vec{x})/2}\,F_1(\t)
 \label{eBack*}
\end{equation}
for some complex number $C$. The exponent of $\ddt$ follows simply from comparing engineering dimensions of the left-hand side and the right-hand side.

Since this sequence of $Q_I$ operators, corresponding to a closed path in the Adinkra, must send bosons to bosons and fermions to fermions, it must be that $\wt(\vec{x})$ is even, so that $\inv2\wt(\vec{x})$ is indeed integral and \Eq{eBack*} is well-defined.

Applying $Q_1^{x_1}\cdots Q_N^{x_N}$ twice to $F_1(\t)$, we find:
\begin{align}
 Q_1^{x_1}\cdots Q_N^{x_N}\cdot Q_1^{x_1}\cdots Q_N^{x_N} F_1(\t)
  &= C^2 \ddt^{2\L}F_1(\t). \label{eCC0}
\intertext{On the left side, using the supersymmetry algebra\eq{eSuSy}, we can anti-commute the $Q_I$ past each other. Rearrange these to regroup the result into}
 &=(-1)^{\wt(\vec{x})\choose2}Q_1^{2x_1}\cdots Q_N^{2x_N}F_1(\t), \label{eCC1}
\intertext{which, using the supersymmetry algebra\eq{eSuSy} again), becomes}
 &=(-1)^{\wt(\vec{x})\choose2}\,i^{\wt(\vec{x})} \,\ddt^{\wt(\vec{x})} F_1(\t).
  \label{eCC2}
\end{align}
Comparing the exponents of $\ddt$ in \Eq{eCC0} with that in \Eq{eCC2} confirms \Eq{eBack*}.
Comparing the numerical coefficients produces:
\begin{align}
 C^2=(-1)^{\wt(\vec{x})\choose2}\,i^{\wt(\vec{x})}
    =i^{\wt(\vec{x})(\wt(\vec{x})-1)+\wt(\vec{x})}
    =i^{\left(\wt(\vec{x})^2\right)}.
\end{align}
Since $\wt(\vec{x})$ is even, we know that $\wt(\vec{x})^2$ is a multiple of 4.  Thus, the left hand side of this is 1.

Using (\ref{eQb}) and (\ref{eQf}) repeatedly, we see that $C$ is $\pm 1$ if $\wt(\vec{x})=0\pmod4$, and $\pm i$ if $\wt(\vec{x})=2\pmod4$. But $C^2=1$ implies that $C=\pm 1$, and this can happen only if $\wt(\vec{x})=0\pmod4$.
\end{proof}

\section{Codes, Again}\label{s:codes-redux}
\subsection{Examples of Doubly-Even Codes}
 \label{s:deC}
Since doubly even codes classify chromotopologies, it is useful to consider a few examples of such codes.  For each $N$ there is a trivial doubly even code $\{00\cdots 0\}$ with one element, which we call $t_N$; its generating set is the empty set. In addition, when $N=4$, there is a code $\{0000,1111\}$, called $d_4$.  The generating set is $\{1111\}$.  More generally, for every even $N\ge 4$, there is a doubly even code called $d_N$, of length $N$ and with $\frac{N}{2}-1$ generators, with generating set
\begin{equation}
\begin{bmatrix}
1\,1\,1\,1\,0\,0\,0\,0\,0\,\cdots\,0\,0\,0\,0\,0\\[-1mm]
0\,0\,1\,1\,1\,1\,0\,0\,0\,\cdots\,0\,0\,0\,0\,0\\[-1mm]
0\,0\,0\,0\,1\,1\,1\,1\,0\,\cdots\,0\,0\,0\,0\,0\\[-1mm]
 \qquad\quad\vdots\\[-1mm]
0\,0\,0\,0\,0\,0\,0\,0\,0\,\cdots\,0\,1\,1\,1\,1
\end{bmatrix}.
\end{equation}
Note that this is a description of the generating set, so that the actual code has more codewords, including the null-vector and all those constructed by adding (bitwise, modulo 2) any number of these generators together. For example,
\begin{equation}
 \begin{bmatrix}
 1\,1\,1\,1\,0\,0\,0\,0\\[-1mm]
 0\,0\,1\,1\,1\,1\,0\,0\\[-1mm]
 0\,0\,0\,0\,1\,1\,1\,1\\[-1mm]
 \end{bmatrix}
\qquad\text{generates}\qquad
 \left\{\begin{array}{rc}
        v_0=&0\,0\,0\,0\,0\,0\,0\,0\\[-1mm]
        v_1=&1\,1\,1\,1\,0\,0\,0\,0\\[-1mm]
        v_2=&0\,0\,1\,1\,1\,1\,0\,0\\[-1mm]
        v_3=&0\,0\,0\,0\,1\,1\,1\,1\\[-1mm]
        v_1\codeplus v_2=&1\,1\,0\,0\,1\,1\,0\,0\\[-1mm]
        v_2\codeplus v_3=&0\,0\,1\,1\,0\,0\,1\,1\\[-1mm]
        v_1\codeplus v_3=&1\,1\,1\,1\,1\,1\,1\,1\\[-1mm]
        v_1\codeplus v_2\codeplus v_3=&1\,1\,0\,0\,0\,0\,1\,1\\[-1mm]
        \end{array}\right\}.
 \label{ed8}
\end{equation}
Note that the same code, on the right-hand side of the display\eq{ed8}, is just as well generated by $\{v_1,v_2,(v_1\codeplus v_3)\}$ and several other choices. For general $N$, the $d_N$ code contains $2^{\frac{N}2-1}$ codewords.

When $N$ is congruent to $7$ or $8$ modulo $8$, there is an important doubly even code called $e_N$, the generating set of which is that of $d_N$ (or $t_1\oplus d_{N-1}$ when $N\equiv 7 \pmod{8})$ augmented by an additional generator of the form $101010\cdots$.  For instance,
\begin{equation}
e_7\,:~
\begin{bmatrix}
1\,1\,1\,1\,0\,0\,0\\[-1mm]
0\,0\,1\,1\,1\,1\,0\\[-1mm]
1\,0\,1\,0\,1\,0\,1
\end{bmatrix},\qquad
e_8\,:~
\begin{bmatrix}
1\,1\,1\,1\,0\,0\,0\,0\\[-1mm]
0\,0\,1\,1\,1\,1\,0\,0\\[-1mm]
0\,0\,0\,0\,1\,1\,1\,1\\[-1mm]
1\,0\,1\,0\,1\,0\,1\,0
\end{bmatrix},
\label{eN}
\end{equation}
and we then write:
\begin{equation}
e_{15}\,:~
\begin{bmatrix}
1\,1\,1\,1\,0\,0\,0\,0\,0\,0\,0\,0\,0\,0\,0\\[-1mm]
0\,0\,1\,1\,1\,1\,0\,0\,0\,0\,0\,0\,0\,0\,0\\[-1mm]
0\,0\,0\,0\,1\,1\,1\,1\,0\,0\,0\,0\,0\,0\,0\\[-1mm]
0\,0\,0\,0\,0\,0\,1\,1\,1\,1\,0\,0\,0\,0\,0\\[-1mm]
0\,0\,0\,0\,0\,0\,0\,0\,1\,1\,1\,1\,0\,0\,0\\[-1mm]
0\,0\,0\,0\,0\,0\,0\,0\,0\,0\,1\,1\,1\,1\,0\\[-1mm]
1\,0\,1\,0\,1\,0\,1\,0\,1\,0\,1\,0\,1\,0\,1
\end{bmatrix},\qquad
e_{16}\,:~
\begin{bmatrix}
1\,1\,1\,1\,0\,0\,0\,0\,0\,0\,0\,0\,0\,0\,0\,0\\[-1mm]
0\,0\,1\,1\,1\,1\,0\,0\,0\,0\,0\,0\,0\,0\,0\,0\\[-1mm]
0\,0\,0\,0\,1\,1\,1\,1\,0\,0\,0\,0\,0\,0\,0\,0\\[-1mm]
0\,0\,0\,0\,0\,0\,1\,1\,1\,1\,0\,0\,0\,0\,0\,0\\[-1mm]
0\,0\,0\,0\,0\,0\,0\,0\,1\,1\,1\,1\,0\,0\,0\,0\\[-1mm]
0\,0\,0\,0\,0\,0\,0\,0\,0\,0\,1\,1\,1\,1\,0\,0\\[-1mm]
0\,0\,0\,0\,0\,0\,0\,0\,0\,0\,0\,0\,1\,1\,1\,1\\[-1mm]
1\,0\,1\,0\,1\,0\,1\,0\,1\,0\,1\,0\,1\,0\,1\,0
\end{bmatrix},
\end{equation}
and so on.

These are famous codes: $e_7$ is known as the Hamming $[7,3]$ code, and $e_8$ is the parity-extended Hamming code. Ref.\cite{rCPS}, describes a so-called ``Construction A'', which determines a lattice as a subset of $\ZZ^N$ of all the points whose coordinates modulo $2$ are in the code, and under this, we form the famous lattices $e_7$ and $e_8$.  The points that are of closest distance to the origin form the root lattice for the Lie algebras $E_7$ and $E_8$, respectively. This correspondence is in fact more general\cite{rNEConstrA}, and may in particular also used to reconstruct the root lattices of the $D_N$ Lie algebras.

Besides the trivial doubly even code $t_N$, $\{000\cdots 0\}$, for any $N\equiv0\pmod4$, there is an $[N,1]$ doubly even code $h_N$ consisting of $\{000\cdots 0,111\cdots 1\}$, the generating set of which is $\{111\cdots 1\}$.\Ft{This is the only code mentioned here not specifically named in Ref.\cite{rCHVP}.} Note that $h_4=d_4$, but $h_N\subset d_N$ for $N=8,12,16,\dots$.

There are many other doubly even codes, and the number grows quickly as $N$ becomes large; see Appendix~\ref{app:RM} and Refs.\cite{rBilRees,rCPS}.

\subsection{Forgetting the Color of Edges, Permutation Equivalence, and $R$-Symmetry}
It is also possible to permute the columns in a code.  For instance, for $e_7$ we might swap the last two columns and obtain a generating set
\begin{equation}
\begin{bmatrix}
1\,1\,1\,1\,0\,0\,0\\
0\,0\,1\,1\,1\,0\,1\\
1\,0\,1\,0\,1\,1\,0
\end{bmatrix}.
\end{equation}
This is another doubly even code, and it is different from $e_7$ as given in \Eq{eN}.  To verify this, one could write the 8 codewords in both cases and compare.
More generally, any permutation of columns of a code will produce another code, which is sometimes the same code, sometimes not.

An example where a column-permutation results in precisely the same code again can be seen by taking the $e_7$ generating set\eq{eN}, and swapping the first and third column, then swapping the second and fourth column.  The result would be:
\begin{equation}
\begin{bmatrix}
1\,1\,1\,1\,0\,0\,0\\
0\,0\,1\,1\,1\,1\,0\\
1\,0\,1\,0\,1\,0\,1
\end{bmatrix}
\quad\longrightarrow\quad
\begin{bmatrix}
1\,1\,1\,1\,0\,0\,0\\
1\,1\,0\,0\,1\,1\,0\\
1\,0\,1\,0\,1\,0\,1
\end{bmatrix}.
\end{equation}
The result does not look like the $e_7$ generating set, but it generates the same code.  Indeed, replace the second generator with the sum of the first and the second generator, and we recover exactly the original generator set for $e_7$.

Yet all the so-obtained codes are in some sense similar, and we call codes related in this way {\em\/permutation-equivalent}.  Note the distinction between code equality and their permutation equivalence, which is weaker than equality.  It is convenient for the classification and naming of codes to give one name for the permutation equivalence class, and recognize the multiplicity of codes that the name represents.

Since the columns of a code correspond to the various $Q_I$, a permutation of the columns of the code corresponds to a permutation of the $Q_I$, \ie, to an $R$-symmetry.  For real $N$-extended supersymmetry, the group of $R$-symmetries is $O(N)$; the permutation equivalences describe the subgroup of this matrix group consisting of permutation matrices.  Though this might suggest that the physically relevant question is permutation equivalence of codes, this is not necessarily so: It may well be possible to construct a theory with two types of supermultiplets, corresponding to two different but permutation-equivalent codes, coupled in a way that precludes rewriting the same theory in terms of only one type of supermultiplet. Although different in technical detail, the inextricable coupling of chiral and twisted-chiral supermultiplets discovered in Ref.\cite{rGHR} is a conceptual paradigm of this possibility.

The columns also correspond to the colors of the Adinkra, so permutation equivalence classes give rise to Adinkra topologies (without the edge colors).  This raises the question: do permutation equivalence classes of doubly even codes classify connected Adinkra topologies?  Certainly we have just described a map from the set of permutation equivalence class of doubly even codes to the set of connected Adinkra topologies.  And certainly this map is surjective.  But is it injective?  That is, is it possible that two non-equivalent doubly even codes will give rise to the same Adinkra topology?  The answer to this question is not clear, but luckily, in trying to classify Adinkras, we can leapfrog the issue of classifying Adinkra topologies and instead use the classification of Adinkra chromotopologies, where the issue is clear.

The set of column-permutations that do not change the (complete) code forms a group, called the automorphism group of the code, $\Aut(G)$.  The number of codes permutation-equivalent to $G$ is then $N!/|\Aut(G)|$, and is regarded the ``mass'' of the code.  Below is a table listing the number of elements in $\Aut(G)$.  Note that $d_4$, $e_7$ and $e_8$ do not fit the pattern for the other $d_N$ or $e_N$.

\begin{table}[tb]
  \centering
\begin{tabular}{l|l}
\boldmath$G$&\boldmath$|{\bf Aut}(G)|$\\\hline
$t_N$&$N!$\\
$d_4$&$24$\\
$d_{2m}$, $m>2$&$2^{m-1}m!$\\
$e_7$&168\\
$e_8$&1344\\
$e_{8m-1}$, $m>1$&$2^{4m-1}(4m-1)!$\\
$e_{8m}$, $m>1$&$2^{4m-1}(4m)!$\\
$h_{4m}$, $m>1$&$(4m)!$\\
\end{tabular}
  \caption{Number of elements in $\Aut(G)$}
  \label{t:known}
\end{table}

It turns out to be possible to independently determine the total number of $[N,k]$ codes of various kinds, including those we need here. This total number must equal the sum of all ``mass''-contributions, $N!/|\Aut(G_i)|$, of all $[N,k]$ codes, $G_i$. Such sum rules are called ``mass formulae'', and
 Gaborit\cite{rPGMass} provides formulae to obtain all that we will need.  These formulae are written by cases according to the congruence class of $N$ modulo 8, and are given in Appendix~\ref{app:RM}.

These numbers, even for moderate $N$ such as $N=11$, are intimidating.  Nevertheless, in some cases, these are due to the many permutation-equivalent codes.  Table~\ref{t:G2} provides a listing of permutation equivalence classes for $N$ up to $11$, obtained by combining the results from Table~\ref{t:known} and accounting for permutations.

Beyond this, however, the number of permutation equivalence classes is still large: see Table~\ref{t:G3} and Appendix~\ref{app:RM} for details.  It is plainly impossible to list all the doubly even $[N,k]$-codes for $N\leq32$ in journal publication.
\begin{table}[ht]
\begin{center}
{\small
\begin{tabular}{l|rr|rr|rr|rr|rr|}
\boldmath$N$&\multicolumn{2}{c|}{\boldmath$k=0$}&\multicolumn{2}{c|}{\boldmath$k=1$}
            &\multicolumn{2}{c|}{\boldmath$k=2$}&\multicolumn{2}{c|}{\boldmath$k=3$}
            &\multicolumn{2}{c|}{\boldmath$k=4$}\\\hline
4&$t_4$&(1)&$d_4$&(1)\\\hline
5&$t_5$&(1)&$t_1\oplus d_4$&(5)\\\hline
6&$t_6$&(1)&$t_2\oplus d_4$&(15)&$d_6$&(15)\\\hline
7&$t_7$&(1)&$t_3\oplus d_4$&(35)&$t_1\oplus d_6$&(105)&$e_7$&(30)\\\hline
8&$t_8$&(1)&$t_4\oplus d_4$&(70)&$t_2\oplus d_6$&(420)&$t_1\oplus e_7$&(240)&
$e_8$&(30)\\
&&&$h_8$&(1)&$d_4\oplus d_4$&(35)&$d_8$&(105)&&\\\hline
9&$t_9$&(1)&$t_5\oplus d_4$&(126)&$t_3\oplus d_6$&(1260)&$t_2\oplus e_7$&(1080)&
$t_1\oplus e_8$&(270)\\
&&&$t_1\oplus h_8$&(9)&$t_1\oplus d_4\oplus d_4$&(315)&$t_1\oplus d_8$&(945)&&\\\hline
10&$t_{10}$&(1)&$t_6\oplus d_4$&(210)&$t_4\oplus d_6$&(3150)&$t_3\oplus e_7$&(3600)&$t_2\oplus e_8$&(1350)\\
&&&$t_2\oplus h_8$&(45)&$t_6 * d_6$&(630)&$d_4\oplus d_6$&(3150)&$d_{10}$&(945)\\
&&&&&$t_2\oplus d_4\oplus d_4$&(1575)&$t_2\oplus d_8$&(4725)&&\\\hline
11&$t_{11}$&(1)&$t_7\oplus d_4$&(330)&$t_5\oplus d_6$&(6930)&$t_4\oplus e_7$&(9900)&
$t_3\oplus e_8$&(4950)\\
&&&$t_3\oplus h_8$&(165)&$t_1 \oplus t_6 * d_4$&(6930)&$t_1\oplus d_4\oplus d_6$&(34650)&$t_1\oplus d_{10}$&(10395)\\
&&&&&$t_3\oplus d_4\oplus d_4$&(5775)&$t_3\oplus d_8$&(17325)&$d_4\oplus e_7$&(9900)\\
&&&&&&&$t_5*d_6$&(13860)&&\\\hline
\end{tabular}}
\end{center}
\caption{A listing of permutation equivalence classes for $N$ up to $11$, with the number of codes in the permutation equivalence class given in in the parentheses. Here, $\oplus$ denotes a vector space direct sum, so that if $U \subset (\ZZ_{2})^{N}$ and $V \subset (\ZZ_{2})^{M}$, then $U\oplus V \subset (\ZZ_{2})^{N} \oplus (\ZZ_{2})^{M}\cong (\ZZ_{2})^{N+M}$. The notation $t_{M}*C$ denotes the direct sum together with at least one additional ``glue'' codeword extending into the $t_{M}$ summand, similar to how a $e_{N}$ code is constructed from the corresponding $d_{N}$ code.}
\label{t:G2}
\end{table}
\begin{table}[ht]
\begin{center}\scriptsize
\begin{tabular}{r|r|r|r|r|r|r|r|r|r|r|r|r|r|r|r|r}
\hbox to0pt{\hss\small\boldmath$_N\!\backslash\!^k$}\kern-2pt
 &\bf1&\bf2&\bf3&\bf4&\bf5&\bf6&\bf7&\bf8&\bf9&\bf10&\bf11&\bf12&\bf13&\bf14&\bf15&\bf16\\[1mm]\hline\hline
 \bf 4 & 1\\\hline
 \bf 5 & 1\\\hline
 \bf 6 & 1& 1\\\hline
 \bf 7 & 1& 1& 1\\\hline
 \bf 8 & 2& 2& 2& 1\\\hline
 \bf 9 & 2& 2& 2& 1\\\hline
 \bf10 & 2& 3& 3& 2\\\hline
 \bf11 & 2& 3& 4& 3\\\hline
 \bf12 & 3& 5& 7& 7& 2\\\hline
 \bf13 & 3& 5& 8& 8& 4\\\hline
 \bf14 & 3& 7& 12& 14& 9& 4\\\hline
 \bf15 & 3& 7& 15& 20& 15& 8& 2\\\hline
 \bf16 & 4& 10& 23& 38& 36& 23& 9& 2\\\hline
 \bf17 & 4& 10& 25& 45& 50& 34& 14& 3\\\hline
 \bf18 & 4& 13& 34& 72& 94& 79& 35& 9\\\hline
 \bf19 & 4& 13& 40& 94& 146& 141& 75& 19\\\hline
 \bf20 & 5& 17& 57& 158& 295& 353& 231& 84& 10\\\hline
 \bf21 & 5& 17& 63& 194& 439& 629& 494& 198& 38\\\hline
 \bf22 & 5& 21& 83& 298& 812& 1481& 1465& 740& 187& 25\\\hline
 \bf23 & 5& 21& 95& 387& 1287& 2970& 3811& 2362& 714& 119& 11\\\hline
 \bf24 & 6& 27& 129& 607& 2444& 7287& 12395& 10048 & 3710 & 739 & 94 & 9\\\hline
 \bf25 & 6& 27& 141& 755& 3808& 15177& 35916 & 38049 & 16039 & 2973 & 309 & 22 \\\hline
 \bf26 & 6& 32& 180& 1114& 6923& 37455& 128270& 194626& 103527& 20206& 1829& 103 \\\hline
 \bf27 & 6& 32& 202& 1435& 11320& 86845 & * & * & * & 174809& 13578& 525 \\\hline
 \bf28 & 7& 39& 263& 2136& 20812& * & * & * & * & * & * & 7402& 151 \\\hline
 \bf29 & 7& 39& 287& 2693& 34233& * & * & * & * & * & * & * & 1940 \\\hline
 \bf30 & 7& 46& 359& 3866& * & * & * & * & * & * & * & * & * & 731 \\\hline
 \bf31 & 7& 46& 400& 4972& * & * & * & * & * & * & * & * & * & * & 210 \\\hline
 \bf32 & 8& 55& 506& * & * & * & * & * & * & * & * & * & * & * & * & 85 \\\hline\hline
\hbox to0pt{\hss\small\boldmath$^N\!/\!_k$}\kern-2pt\vrule width0pt height3ex
 &\bf1&\bf2&\bf3&\bf4&\bf5&\bf6&\bf7&\bf8&\bf9&\bf10&\bf11&\bf12&\bf13&\bf14&\bf15&\bf16\\\end{tabular}
\end{center}
\caption{Number of distinct permutation classes of doubly even $[N,k]$ codes. The ``\,*\,'' entry indicates codes that are still being enumerated; see {\small\tt http://www.rlmiller.org/de\_codes/} for up-to-date results, including links to listings of the actual codes.}
\label{t:G3}
\end{table}

\subsection{Coset Enumerators and One-Hook Hanging Adinkras}
Since we have concluded that any connected Adinkra chromotopology is a quotient of the $N$-cube by a doubly even $[N,k]$-code $G$, the vertices of the Adinkra (that is, the component fields of the supermultiplet) correspond to cosets of $\{0,1\}^N$, thought of as $(\ZZ_2)^N$, by the subgroup $G$. This immediately implies that the Adinkra has $2^{N-k}$ nodes, where $k$ is the dimension of $G$. Of these, one half represents bosonic component fields, and the other half fermionic component fields in the corresponding supermultiplet.  We thus have that the number of bosonic component fields, $d_B$, and the number of fermionic component fields, $d_F$, after taking this quotient satisfy 
\begin{equation}
d_B ~=~ d_F ~=~  2^{N - k - 1}.
 \label{eDBF1}
\end{equation}

In Ref.\cite{r6-1}, we described a notion of hanging a graph by a one or more sources.  For instance, if we pick the vertex $v_*$ and hang the graph by it, we let all the arrows on edges point from the vertices that are further away from $v_*$ (as measured through the edge set) to vertices that are closer to $v_*$.  Equivalently, for each vertex $v$ we define the engineering degree of $v$ to be
\begin{equation}
 [v]=[v_*]-\frac{1}{2}\mbox{dist}(v,v_*),
 \label{v*v}
\end{equation}
where $\mbox{dist}(v,v_*)$ is the length of the shortest path from $v$ to $v_*$ in the edge set.  Then the arrows are drawn in the direction of increasing engineering degree, upward. 

For this, ``one-hooked'' Adinkra, we can examine how many component fields are in each engineering degree.  By \Eq{v*v}, this is equivalent to finding how many vertices are of a given distance from the highest vertex, $v_*$.  The corresponding notion in coding theory is the ``coset weight enumerator''\cite{rCHVP}.  This is a polynomial of the form $\sum_\ell a_\ell\, x^\ell$ where $a_\ell$ is the number of cosets whose distance to the $0$ coset is $\ell$.  Thus, the coset weight enumerator for a doubly even code can be used to find the number of component fields in each engineering dimension for a one-hooked Adinkra corresponding to that doubly even code.

The equation of\eq{eDBF1} implies that in order to minimize the numbers of bosonic and fermionic fields in an Adinkra, one must maximize the value of $k$ to $k=\vk(N)$. Thus the minimum number of bosons, $\min(d_B)$, and the minimum number of fermions, $\min(d_F)$, satisfy
\begin{equation}
\min(d_B) ~=~ \min(d_F) ~=~  2^{N - \vk(N) -1 }.
 \label{eDBF2}
\end{equation}
Since the ultimate application toward which all our efforts are aimed is a comprehensive 
understanding of the representation theory of spacetime supersymmetry in higher dimensions,
 it is perhaps useful to use this formula to gains some insights into the representations present in 4D,
$\cal N$-extended supersymmetric theories (note: ${\cal N}\Defl N/4$).  For this purpose we
note the results in Table~\ref{t:1cN8}, calculated using\eq{eDBF2}.  
\begin{table}[tb]
  \centering
\begin{tabular}{r|r|r}
\boldmath$N$&\boldmath$\cal N$&\boldmath$\min(d_B)$\\\hline
4&1&$4$\\
8&2&$8$\\
12&3&$64$\\
16&4&$128$\\
20&5&$1,024$\\
24&6&$2,048$\\
28&7&$16,384$\\
32&8&$32,768$\\
\end{tabular}
\caption{Minimal numbers for $d_B$ and $d_F$}
\label{t:1cN8}
\end{table}

The lowest values are in perfect agreement with well known results. The value for ${\cal N}=1$ 
corresponds to the chiral supermultiplet, the value for ${\cal N}=2$ corresponds to the vector and tensor supermultiplets, and the value of ${\cal N}=3,4$ corresponds to the respective Weyl supermultiplets.  Curiously enough and by very different arguments\cite{rGLPR}, the final result of $\min(d_B)=\min(d_F)=32,768$ in this table appeared very early in our considerations of iso-spinning particles  and ``Garden Algebras'' as the spectrum generating algebras of spacetime supersymmetry.

\section{Toward a Classification of Adinkras and Supermultiplets}
 \label{CliffHanger}
We have in this paper described the classification of chromotopologies of Adinkras: they are disjoint unions of connected chromotopologies, each of which is described by a doubly even binary linear code. The converse fact that every doubly even code can be used to form an Adinkra will be proved in Ref.\cite{r6-3.2}.  This is done by taking an arbitrary doubly even code and constructing a supermultiplet whose Adinkra chromotopology comes from the code.

We have already addressed the issue of orienting the edges in Ref.\cite{r6-1}: given a chromotopology, consider all the ways to hang the vertices of the graph at various heights, subject to certain conditions. In  Ref.\cite{r6-1} we have also proven that all such height assignments can be obtained one from another through a sequence of `vertex-raising/lowering'. In turn, the possible choices in dashing of edges, \ie, assigning factors of $-1$ to the action of certain $Q_I$'s on certain component fields (see Table~\ref{t:A}), is closely related to the classification of representations of the Clifford algebras\cite{rLM} and their iterated $\ZZ_2$-quotients, and is deferred to subsequent work.

As per Definition~\ref{dA}, a choice of a chromotopology together with a choice of orientations for all the edges and a choice of their dashedness fully specifies every Adinkra.

In light of all of the above, does this classify one-dimensional $N$-extended off-shell supermultiplets?  Not quite.  First, there is the question of whether every such supermultiplet comes from an Adinkra.  Our investigations in this area indicate that there do indeed exist non-Adinkraic supermultiplets, but Adinkras turn out to be useful in their description.  Second, we might ask whether two Adinkras may describe the same supermultiplet.  This indeed occurs, and this issue will be settled in Ref.\cite{r6-3.2}.  Suffice it to say here that, depending on the engineering 
dimensions of the component fields, distinct Adinkras usually correspond to distinct supermultiplets. Recall also that the number of possible such choices of engineering dimensions grows combinatorially with the number of component fields\cite{r6-1}, which in turn grows exponentially with $N$.

We are now also in position to review a minor controversy from a more complete perspective.  There is in the literature a work\cite{rKRT} the title of which leaves the impression of the existence of a prior complete classification of all off-shell representations of $N$-extended worldline supersymmetry without central charges.  Ref.\cite{r6-3a} then shows that only using the scheme described in\cite{rKRT}, there exist counter-examples to this impression, which in turn caused two of the Authors of\cite{rKRT} to create a further `refinement'\cite{rKT07} to their scheme.

Our current work shows that the problem of a complete classification of off-shell representations of $N$-extended worldline supersymmetry (and also Adinkras) includes the simpler problem of a complete classification of Adinkra chromotopologies.  The latter problem is itself a formidable one as apparently the number of inequivalent chromotopologies grows combinatorially with $N$ and $k$; see Table~\ref{t:G3}.  The works of Refs.\cite{rKRT,rKT07} include neither reference nor indication to this simpler problem of classifying chromotopologies or some equivalent thereof, and with it the emergence of the role played by doubly even binary linear error-correcting codes.
 Even with the proposed `refinement' the works of Refs.\cite{rKRT,rKT07} must then be regarded as leaving more work necessary for the presentation of a complete theory of the off-shell representations of $N$-extended worldline supersymmetry.

Finally, we wish to draw attention to the ``degeneracy'' in constructing even the minimal supermultiplets  for various $N$, uncovered by the listing of permutation equivalence classes of doubly even binary linear codes in Table~\ref{t:G3}.

 In Table~\ref{t:1N16}, we list the dimensions of the minimal representations for $1\leq N\leq16$, and provide the doubly even binary linear codes which encode the projection of $I^N=\{0,1\}^N$ so as to obtain the corresponding Adinkra topologies. Starting with $N=10$, there is more than one such permutation-inequivalent projection, and so more than one such minimal Adinkra topology.
 Each of these corresponds to at least one minimal supermultiplet, obtained by iteratively projecting an ``unprojected'' supermultiplet such as the Clifford supermultiplet, $\sM^=_{I^N}$, defined in section~\ref{s:CM}, according to the provided permutation equivalence class of codes. ``Unprojected'' supermultiplets other than $\sM^=_{I^N}$ can, in general, be quotiented only by a subcode of the listed codes, depending on the distribution of the component fields' engineering dimensions.
 
\begin{table}[tb]
  \centering
 {\begin{tabular}{@{} c|c|c|c|c|c|c|c @{}}
    \toprule
  \BM{N=1} &\BM{N=2} &\BM{N=3} &\BM{N=4} &\BM{N=5} &\BM{N=6} &\BM{N=7} &\BM{N=8} \\
  $(1|1)$ &$(2|2)$ &$(4|4)$ &$(4|4)$ &$(8|8)$ &$(8|8)$ &$(8|8)$ &$(8|8)$ \\
    \midrule
  &
  &
  &\stk{{\footnotesize$d_4$}\\
        1111}
  &\stk{{\footnotesize$t_1\8d_4$}\\
        01111}
  &\stk{{\footnotesize$d_6$}\\
        001111\\
        111100}
  &\stk{{\footnotesize$e_7$}\\
        0001111\\
        0111100\\
        1010101}
  &\stk{{\footnotesize$e_8$}\\
        00001111\\
        00111100\\
        11110000\\
        01010101}\\
    \bottomrule
  \omit\\[2mm]
    \toprule
  \BM{N=9} & \BM{N=10} & \BM{N=11} & \BM{N=12}
 & \BM{N=13} & \BM{N=14} & \BM{N=15} & \BM{N=16} \\
  $(16|16)$ & $(32|32)$ & $(64|64)$ & $(64|64)$
 & $(128|128)$ & $(128|128)$ & $(128|128)$ & $(128|128)$ \\
    \midrule
   \stk{{\footnotesize$t_1\8e_8$}\\
        0~00001111\\
        0~00111100\\
        0~11110000\\
        0~01010101}
  &\stk{{\footnotesize$t_2\8e_8$}\\
        00~00001111\\
        00~00111100\\
        00~11110000\\
        00~01010101}
  &\stk{{\footnotesize$t_3\8e_8$}\\
        000~00001111\\
        000~00111100\\
        000~11110000\\
        000~01010101}
  &\stk{{\footnotesize$d_4\8e_8$}\\
        0000~00001111\\
        0000~00111100\\
        0000~11110000\\
        0000~01010101\\
        1111~00000000}
  &\stk{{\footnotesize$t_1\8d_4\8e_8$}\\
        00000~00001111\\
        00000~00111100\\
        00000~11110000\\
        00000~01010101\\
        01111~00000000}
  &\stk{{\footnotesize$d_6\8e_8$}\\
        000000~00001111\\
        000000~00111100\\
        000000~11110000\\
        000000~01010101\\
        001111~00000000\\
        111100~00000000}
  &\stk{{\footnotesize$e_7\8e_8$}\\
        0000000~00001111\\
        0000000~00111100\\
        0000000~11110000\\
        0000000~01010101\\
        0001111~00000000\\
        0111100~00000000\\
        1010101~00000000}
  &\stk{{\footnotesize$e_8\8e_8$}\\
        00000000~00001111\\
        00000000~00111100\\
        00000000~11110000\\
        00000000~01010101\\
        00001111~00000000\\
        00111100~00000000\\
        11110000~00000000\\
        01010101~00000000}\\
    \midrule
  &\stk{{\footnotesize$d_{10}$}\\
        0000001111\\
        0000111100\\
        0011110000\\
        1111000000}
  &\stk{{\footnotesize$t_1\8d_{10}$}\\
        0~0000001111\\
        0~0000111100\\
        0~0011110000\\
        0~1111000000}
  &\stk{{\footnotesize$d_{12}$}\\
        000000001111\\
        000000111100\\
        000011110000\\
        001111000000\\
        111100000000}
  &\stk{{\footnotesize$t_1\8d_{12}$}\\
        0~000000001111\\
        0~000000111100\\
        0~000011110000\\
        0~001111000000\\
        0~111100000000}
  &\stk{{\footnotesize$d_{14}$}\\
        00000000001111\\
        00000000111100\\
        00000011110000\\
        00001111000000\\
        00111100000000\\
        11110000000000}
  &\stk{{\footnotesize$e_{15}$}\\
        000000000001111\\
        000000000111100\\
        000000011110000\\
        000001111000000\\
        000111100000000\\
        011110000000000\\
        101010101010101}
  &\stk{{\footnotesize$e_{16}$}\\
        0000000000001111\\
        0000000000111100\\
        0000000011110000\\
        0000001111000000\\
        0000111100000000\\
        0011110000000000\\
        1111000000000000\\
        1010101010101010} \\
    \midrule
  &
  &\stk{{\footnotesize$d_4\8e_7$}\\
        0000000~1111\\
        0001111~0000\\
        0111100~0000\\
        1010101~0000}
  &
  &\stk{{\footnotesize$d_6\8e_7$}\\
        0000000~001111\\
        0000000~111100\\
        0001111~000000\\
        0111100~000000\\
        1010101~000000}
  &\stk{{\footnotesize$e_7\8e_7$}\\
        0000000~0001111\\
        0000000~0111100\\
        0000000~1010101\\
        0011110~0000000\\
        1111000~0000000\\
        1010101~0000000}
  & \\
    \midrule
  &
  &
  &
  &\stk{{\footnotesize$e_{13}$}\\
        0000000001111\\
        0000000111100\\
        0000011110000\\
        0001111000000\\
        1110101010101}
  &\stk{{\footnotesize$e_{14}$}\\
        00000000001111\\
        00000000111100\\
        00000011110000\\
        00001111000000\\
        00111100000000\\
        11010101010101}
  & \\
    \bottomrule
  \end{tabular}}
  \caption{Minimal supermultiplets of the indicated $N$-extended supersymmetry, with the stated number of (bosonic$|$fermionic) component fields and a generator set for the code, $C$, specifying the chromotopology as the $\{0,1\}^N/C$ quotient.}
  \label{t:1N16}
\end{table}

Notice that Table~\ref{t:1N16} details the cases represented only by the $N\leq16$ right-most entries in Table~\ref{t:G3}, and the corresponding entries in Table~\ref{t:1cN8}. The remaining, $16<N\leq32$ right-most entries in Table~\ref{t:G3} also correspond to topologically distinct minimal supermultiplets, but their sheer number prevents us from displaying them in the manner of Table~\ref{t:1N16}.
 This clearly illustrates that this ``degeneracy'' amongst even the minimal supermultiplets grows extremely fast with $N$.
 
 In turn, this provides a surprising wealth of building blocks for models with $N$-extended worldline supersymmetry as $N$ is increased towards $N=32$.

\begin{flushright}\sl
  Mathematics is not a careful march down a well-cleared highway,\\*
  but a journey into a strange wilderness,\\*
  where explorers often get lost.\\*
   --~W.S.~Anglin
\end{flushright}

\bigskip\bigskip\paragraph{\bf Acknowledgments:}
 The research of S.J.G.\ is supported in part by the National Science
Foundation Grant PHY-0354401.
 T.H.\ is indebted to the generous support of the Department of
Energy through the grant DE-FG02-94ER-40854. The work described in Appendix~\ref{app:RM} uses computer facilities supported by the National Science Foundation Grant No.~DMS-0555776.

\clearpage
\appendix
\section{Deferred Details on Supersymmetry and Adinkras}
 \label{app:SA}
\subsection{Real Coefficients}
\label{app:real}
Real supermultiplets, $\sM=(F_1\6(\t),\cdots,F_m\6(\t))$ consist of real component fields: $\big(F_A(\t)\big)^\dag=F_A(\t)$, and supersymmetry is assumed to preserve this condition.  In this section, we will show that this condition implies that the coefficients $c$ in \Eq{eQb} and \Eq{eQf} are real.

We use the convention whereby $(XY)^\dag=Y^\dag X^\dag$, regardless whether $X$ and $Y$ are bosonic (commuting) or fermionic (anticommuting) objects, as is standard in the physics literature. 

If $\f_A(\t)$ is real, then its supersymmetry transform must also be real. Applying this to \Eq{eQb} results in
\begin{equation}
 \d_Q(\e)\, \f_A(\t) = -i\e^I\,c\,\ddt^{\,[\f_A]+\frac12-[\j_B]} \j_B(\t).
 \label{edQe}
\end{equation}
 Thus
\begin{equation}
 \Big(\d_Q(\e)\,\f_A(\t)\Big)^\dag
 ~=~i\,\ddt^{\,[\f_A]+\frac12-[\j_B]} \j^\dag_B(\t)\,c^*\,\e^I
 ~=~-i\e^Ic^*\,\ddt^{\,[\f_A]+\frac12-[\j_B]} \j_B(\t).\label{edQe+}
\end{equation}
 Comparing the right-hand sides of Eqs.\eq{edQe} and\eq{edQe+}, we find that
\begin{equation}
c^* =\,c.
\label{eRealC}
\end{equation}%
Thus the coefficients $c$ are real.

\subsection{The Proof that Scaling Factors $c=\pm 1$ are the Only Ones Necessary}
\label{app:pm}
We now will show that the coefficients in an adinkraic supermultiplet can be chosen to be $c=\pm 1$ via a rescaling of fields by real numbers.

\begin{proposition}\label{rescale}
Suppose we have an adinkraic supermultiplet, with $F_1\6(\t),\cdots,F_{2m}\6(\t)$ for component fields.  There is a real rescaling of these component fields so that each non-zero coefficient $c$ in Eqs.\eq{eQb} and\eq{eQf} is equal to $1$ or $-1$.
\end{proposition}

\begin{proof}
We suppose we have component fields $F_1\6(\t),\cdots,F_{2m}\6(\t)$,\footnote{See footnote~\ref{introf} in the proof of Theorem~\ref{T:QuoC} about the use of this notation.} and supersymmetry generators $Q_I$ so that the supersymmetry transformation rules are all of the form
\begin{equation}
Q_I F_A = c\,\ddt^\lambda F_B\label{adinkralike}
\end{equation}
where $\lambda$ is either $0$ or $1$, and $F_B$ is another component field, and $c$ is a complex number (real if $F_A$ is bosonic, pure imaginary if $F_A$ is fermionic).

We can draw an Adinkra-like graph for this supermultiplet: we again create vertices $v_1,\cdots,v_{2m}$ to correspond to the component fields $F_1,\cdots,F_{2m}$---white for bosonic fields and black for fermionic fields.  For each supersymmetry transformation like the one above, we draw an edge from the vertex $v_A$ to $v_B$.  The arrow goes from $v_A$ to $v_B$ if $\lambda=0$, and the reverse if $\lambda=1$.  As before, this is consistent with the corresponding equation for $Q_I F_B$.  Instead of choosing between a dashed or solid line, we label the edge by $c$ if $F_A$ is bosonic and $i/c$ if it is fermionic.  It is easy to check that this is consistent with the corresponding equation for $Q_I F_B$.

We can work on each connected component of this Adinkra-like graph separately.  So for the following we assume our Adinkra-like graph is connected.

We wish to rescale the various component fields so that $c$ has absolute value 1.  A naive approach would be to simply rescale $F_B$ by $c$, so that $c$ disappears in \Eq{adinkralike}.  But remember that there are $2mN$ of these equations, and only $2m$ component fields.  More visually, imagine the Adinkra-like graph with separate $c$-labels on each edge.  We can imagine starting at one vertex, then going along an edge to the next vertex, rescaling the corresponding field by whatever it takes to make $c=1$ on that edge.  The trouble is that when we come back to a vertex we already saw, we are no longer free to rescale that field without messing up other $c$'s.  But as we will see, it turns out that the supersymmetry algebra guarantees that when we return to a previously-seen vertex, the edge will have $|c|=1$.

We first begin by removing all the loops, using the following standard procedure from graph theory: if we choose an edge, and erase it from the graph, the graph may either separate or stay connected.  We first find an edge so that erasing it keeps the graph connected.  By successively finding such edges and erasing them, we obtain a minimal graph so that this graph is still connected.  This resulting graph is a tree, so that for every pair of vertices $v_i$ and $v_j$, there exists a unique, non-backtracking path from $v_i$ to $v_j$ in this tree (were it not unique, we could remove another edge).  This tree is called a {\em\/spanning tree} for the original graph.  The spanning tree for a graph is not unique, but we select one.  See Figure~\ref{fig:spantree}.

\begin{figure}[ht]
\begin{center}
\begin{picture}(130,75)(0,0)
 \put(1,2){\includegraphics{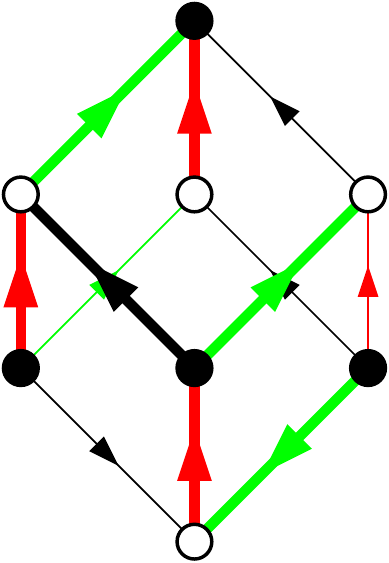}}
 \put(80,2){\includegraphics{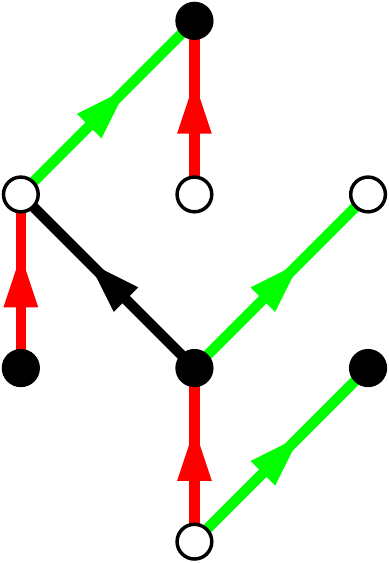}}
\end{picture}
\end{center}
\caption{On the left is an $N=3$ cubical Adinkra.  The dashedness of the lines are suppressed for clarity.  Some edges (shown with thin lines) can be deleted, until what is left is a tree, called the {\em\/spanning tree} (shown at right).\label{fig:spantree}}
\end{figure}

Now pick any vertex in our graph, denote it $v_*$ and the corresponding component field $F_*(\t)$.  For simplicity of exposition we can assume $v_*$ is bosonic.  If $v_i$ is any other vertex, there is a sequence of edges in the tree connecting $v_*$ with $v_i$, corresponding to a sequence of $Q_I$'s, say $Q_{I_1},\cdots, Q_{I_n}$.  We then apply \Eq{adinkralike} iteratively to $F_*(\t)$, obtaining
\begin{equation}
 Q_{I_n}\cdots Q_{I_1} F_*(\t)
  = C_i\,\ddt^{\lambda} F_i(\t) \label{treeAd}
\end{equation}
for some non-negative integer $\lambda$ (recording the number of arrows that point in the opposite direction of the path), and some non-zero coefficient $C_i$.  This coefficient $C_i$ is real if $n\equiv 0, 1\pmod{4}$, and pure imaginary if $n\equiv 2, 3\pmod{4}$.  It is then possible to rescale $F_i(\t)$ as
\begin{equation}
\tilde{F}_i(\t):= |C_i|\,F_i(\t),\label{eqn:rescale}
\end{equation}
and using $\tilde{F}_i(\t)$ instead of $F_i(\t)$ in the description of the supermultiplet. This rescaling is by a real number. Do so for every $v_i\neq v_*$. We have now done all the rescaling of the fields that we will do. The result is a new Adinkra-like graph, where each $C_i$ has absolute value 1.

If instead we want to go from $v_i$ back to $v_*$ in the tree, we can do the supersymmetry transformations in reverse:
\begin{equation}
Q_{I_1}\cdots Q_{I_n} F_i = \frac{1}{C_i}\,\ddt^{n-\lambda} F_*. \label{treeinvAd}
\end{equation}

Now suppose we choose two vertices $v_i$ and $v_j$ of the graph.  We construct a path $P$ from $v_j$ to $v_i$ by first taking the non-retracing path in the spanning tree from $v_j$ to $v_*$, followed by the non-retracing path in the spanning tree from $v_*$ to $v_i$.  See Figure~\ref{fig:treejtoi}.  Note that $P$ might be partially retracing, if the last few edges of the first path is retraced backward by a beginning segment of the second path.  This retracing is immaterial for us.

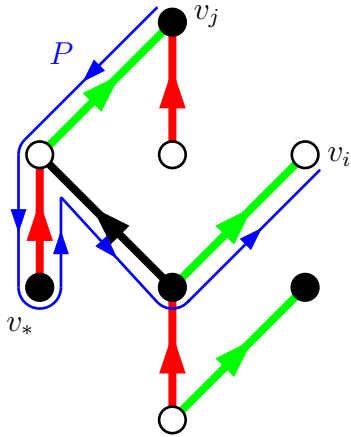
\begin{figure}[ht]
\begin{center}
\begin{picture}(38,50)(-6,5)
 \put(-3.5,-1.5){\includegraphics{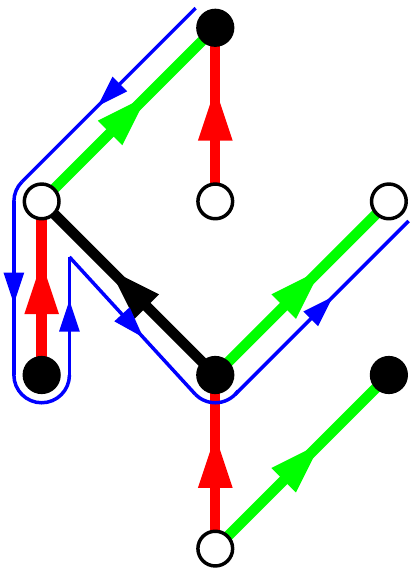}}
 \put(0,12.5){\makebox[0in][r]{$v_*$}}
 \put(39,35){\makebox[0in][l]{$v_i$}}
 \put(23,54){\makebox[0in][c]{$v_j$}}
 \put(2,48){{\color{blue} $P$}}
\end{picture}
\end{center}
\caption{Using the same example as in Figure~\ref{fig:spantree}, suppose we are considering the vertices $v_i$ and $v_j$ as seen in the figure.  The blue path shows a path $P$ in the spanning tree going from $v_j$ to $v_*$, followed by a path in the spanning tree going from $v_*$ to $v_i$.  One edge is retraced, but this is irrelevant for our purposes.  The sequence of colors (black\,=\,1, red\,=\,2, green\,=\,3) along $P$ indicate that we should consider $Q_3Q_1Q_2Q_2Q_3 v_j$ which will be a constant multiple of $\ddt^3 v_i$.  Because of our rescaling in \Eq{eqn:rescale}, we know this constant has absolute value $1$.
\label{fig:treejtoi}}
\end{figure}

If the path $P$ involves the sequence $Q_{J_1}\cdots Q_{J_k}$ of supersymmetry generators, we then use \Eq{treeAd} and \Eq{treeinvAd} to get
\begin{equation}
Q_{J_k}\cdots Q_{J_1} F_j=\frac{C_i}{C_j}\,\ddt^\mu F_i,\label{eqn:treebothAd}
\end{equation}
where $\mu$ is some non-negative integer, measuring the number of arrows pointing against the path $P$.  Note that $C_i/C_j$ has absolute value $1$.

In particular, if we consider any edge of our Adinkra that is in the spanning tree, its label $c$ now has absolute value $1$.

We now examine an edge of the Adinkra that is not in the spanning tree.  Suppose it connects a bosonic vertex $v_i$ to a fermionic vertex $v_j$, \ie, there is a $Q_I$ so that
\begin{equation}
Q_I v_i=c\,\ddt^{\lambda} v_j\label{eqn:newedge}
\end{equation}
for some real $c$ and $\lambda=0$ or $1$.  Now take the path $P$ in the spanning tree described above, namely, the one that goes from $v_j$ to $v_*$ then to $v_i$.  This corresponds to a sequence of supersymmetry generators $Q_{I_1},\cdots,Q_{I_n}$.  We now close the path $P$ into a closed loop, $P'$ by adding to $P$ the new edge going back from $v_i$ to $v_j$.

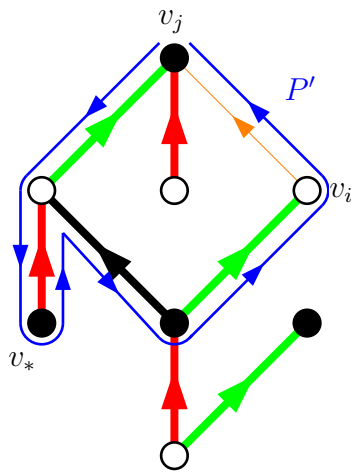
\begin{figure}[ht]
\begin{center}
\begin{picture}(38,60)(-6,-4)
 \put(-3.5,-1.5){\includegraphics{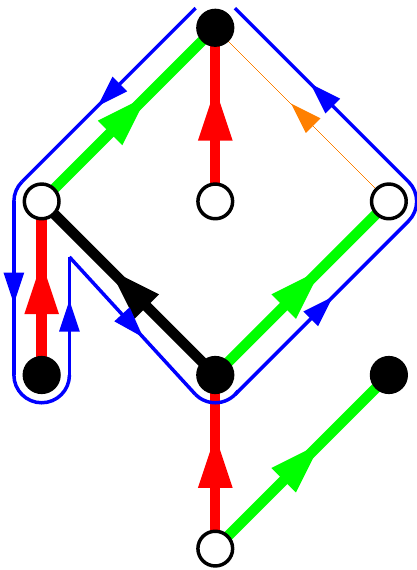}}
 \put(0,12.5){\makebox[0in][r]{$v_*$}}
 \put(39,35){\makebox[0in][l]{$v_i$}}
 \put(18,58){\makebox[0in][c]{$v_j$}}
 \put(33,48){{\color{blue} $P'$}}
\end{picture}
\end{center}
\caption{Suppose we consider the orange edge from $v_i$ to $v_j$, which is in the original Adinkra but not in the spanning tree.  We draw that path $P$ from $v_j$ to $v_i$ in the spanning tree as in Figure~\ref{fig:treejtoi}, and then tack on the orange edge.  The resulting path $P'$, shown in blue, is a loop from $v_j$ to itself.
\label{fig:treeloop}}
\end{figure}

Using \Eq{eqn:newedge} together with \Eq{eqn:treebothAd} we get, for some $r\in\IN$,
\begin{equation}
 Q_I Q_{I_n}\cdots Q_{I_1} F_j(\t)
 = c(-i)^r\frac{C_i}{C_j}\,\ddt^{\lambda+\mu} F_j.
\end{equation}
As this is a closed loop, we can apply this equation to itself, and get
\begin{equation}
 \left(\,Q_I Q_{I_n}\cdots Q_{I_1}\,\right)^2 F_j(\t)
   = (-1)^rc^2\left(\,\frac{C_i}{C_j}\,\right)^2
       \ddt^{2(\lambda+\mu)} F_j(\t).\label{eSqSeq}
\end{equation}
Using \Eq{eSuSy} we can anti-commute the $Q_I$'s past each other, contracting the repeated $Q_I$'s, until we get
\begin{equation}
 \pm (i\ddt)^{\,n+1} F_j(\t)
   = (-1)^r c^2\left(\,\frac{C_i}{C_j}\,\right)^2
       \ddt^{2(\lambda+\mu)} F_j(\t).\label{eResult}
\end{equation}
The $\pm$ sign is determined by $n$ and how many repetitions there are in the $Q_I$ sequence.  Note that $n+1$ is even, since each edge connects vertices of opposite statistics, and a sequence of $n+1$ of these edges go from $v_j$ to itself.  Thus, the left hand side of \Eq{eResult} is actually real.  We can then match coefficients in \Eq{eResult} to get
\begin{equation}
c^2 = \pm (-1)^{r+(n+1)/2}\frac{C_j^2}{C_i^2} = \pm 1.\label{eSign}
\end{equation}
From this we see that $c$ must have complex absolute value $1$.  Since $c$ is real, $c=\pm 1$.
Furthermore, we also note from \Eq{eResult} that $2(\lambda+\mu)=n+1$, indicating that the loop $P'$ involved precisely $(n+1)/2$ arrows going against the path, and thus, $(n+1)/2$ arrows going along the path. This implies that the absence of central charge excludes Escheric loops\cite{rA}.
\end{proof}

This also proves that Adinkraic supermultiplets with no central charge are {\em\/engineerable\/}\cite{r6-1}: it is possible to assign engineering dimensions to all component fields, consistently with the action Eqs.\eq{eQbf} of the supersymmetry algebra\eq{eSuSy}, and without having to introduce parameters of nonzero engineering dimension in either of these.

\section{Complex Adinkras}
\label{app:complex}
Everything in the paper up to now has dealt with supermultiplets thinking of each field as real.  If we wish to represent complex supermultiplets, we can introduce a kind of Adinkra that is analogous to the Adinkras we have been considering up to now, only using vertices to describe complex fields instead of real fields.  Many of the same ideas continue to hold, except that instead of the coefficients $c$ being real, they might also be complex.

If an adinkraic supermultiplet is complex, then the proof of Proposition~\ref{rescale} in Appendix~\ref{app:pm} can be modified in the following way: when the fields are rescaled by a real number to make $|C_i|=1$, we can actually rescale the fields by a complex number to make each $C_i=1$.  Then, \Eq{eSign} shows that the $c^2=\pm 1$, so that $c$ could be $\pm 1$, or it could also be $\pm i$.  So instead of edges being either solid or dashed, there should be two other ways to decorate the edge to denote the additional cases of $i$ and $-i$.

For examples of complex Adinkras, we can take the cubical Adinkras, and simply reinterpret the vertices as representing complex component fields, instead of real ones.  Thus, for $N=1$, there is the base Adinkra and the Klein-flipped base Adinkra, just as in the real case.  Note that these have edges labelled only by $\pm 1$, so it is no surprise that there are other complex Adinkras that are not described in this way.

As in the real case, connected Adinkras can be obtained by taking the quotient of a cube by a binary linear block code (the same proof applies), but this time the condition that the code is doubly even is replaced by the condition that the code is even and self-orthogonal.  That is, codewords need not have weights that are multiple of $4$; rather, we require simply that their weights be even. In this case, we must also impose self-orthogonality, since unlike doubly even codes, even codes are not always self-orthogonal.

For instance, the $N=2$ exterior Adinkra splits.  To be specific, the $N=2$ exterior supermultiplet involves two bosons $\f_1$, $\f_2$ and two fermions $\j_1$, $\j_2$.
\begin{align}
Q_1\f_1&=\4\j_1,&       Q_2\f_1&=-\j_2,\\
Q_1\f_2&=\4\j_2,&       Q_2\f_2&=\4\j_1,\\
Q_1\j_1&=-i\,\ddt\f_1,& Q_2\j_1&=-i\,\ddt\f_2,\\
Q_1\j_2&=-i\,\ddt\f_2,& Q_2\j_2&=\4i\,\ddt\f_1.
\end{align}
To project out, we define $\f=\f_1+i\f_2$, $\bar\f:=\f_1-i\f_2$, $\j=\j_1+i\j_2$, and $\bar\j:=\j_1-i\j_2$.  The supermultiplet $(\f_1,\f_2|\j_1,\j_2)$ may then be rewritten as a complex supermultiplet, $(\f|\j)$, and its conjugate, $(\bar\f|\bar\j)$, with the supersymmetry transformations:
\begin{align}
Q_1\f&=\4\j,&                Q_2\f&=\4i\j,\\
Q_1\j&=-i\,\ddt\f,&          Q_2\j&=-\ddt\f;\\[2mm]
Q_1\bar\f&=\4\bar\j,&        Q_2\bar\f&=-i\bar\j,\\
Q_1\bar\j&=-i\,\ddt\bar\f,&  Q_2\bar\j&=\4\ddt\f.
\end{align}
The (complex) $\IC$-Adinkras for $(\f|\j)$ and for $(\bar\f|\bar\j)$ exhibit a relationship between the two complex supermultiplets:
\begin{equation}
 \vC{
 \begin{picture}(120,35)(-25,-2)
 \put(-2,0){\includegraphics[height=30mm]{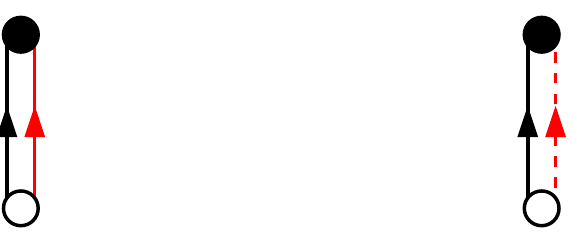}}
 \put(3,28){$\j$}
 \put(4,15){\color{red}$i$}
 \put(-5,15){$1$}
 \put(3,1){$\f$}
 \put(10,15){\vector(1,0){45}}
 \put(10,15){\vector(-1,0){0}}
 \put(25,18){complex}
 \put(22,10){conjugation}
 \put(59,28){$\bar\j$}
 \put(67,15){\color{red}$-i$}
 \put(59,15){$1$}
 \put(59,1){$\bar\f$}
 \end{picture}}
 \label{eCApair}
\end{equation}
The two disagree only on the values of $c$. The unlabeled topologies are the same, with a double edge between two vertices. 

We note that one of these two complex supermultiplets\eq{eCApair} can be regarded as the ``chiral'' supermultiplet, while the other is its conjugate, the ``anti-chiral'' supermultiplet. Also, $Q_2\neq\pm i\,Q_1$; we note that $Q_2=+iQ_1$ when acting on $\f(\t)$ and $\bar\j(\t)$, but $Q_2=-iQ_1$ when acting upon $\bar\f(\t)$ and $\j(\t)$. Thus, $Q_1$ and $Q_2$ really are independent.

\section{Computing Doubly Even Codes}
 \label{app:RM}
 For some classification problems, there are theorems that answer the question once and for all.  For example, the theorem on the classification of finitely generated abelian groups describes the isomorphism class of such a group via a simple sequence of integers.  In most cases in combinatorics, one cannot hope for such a theorem.  Instead one must settle for an exhaustive enumeration up to a certain size.  It turns out that there are more than 1.1 trillion doubly even, binary $[32,k]$ codes, which without compression will take some number of terabytes to store.  Using methods including those described in this appendix, we have already computed over 60,000 of these codes.
 
 The first section of this appendix gives the formula describing the number of codes we are interested in and a lower bound for the number of permutation-isomorphism classes of codes.  The next section explains the computations necessary for computing the automorphism group and canonical representative of a binary code, which are critical for the exhaustive enumeration.  In the third section, some optimizations for the particular situation at hand are discussed, which are used in an implementation written specifically for this purpose. The fourth section describes the overarching algorithm designed to produce the intended enumeration.
 
 The methods used to tackle this problem, namely partition refinement and canonical augmentation, were originally developed by McKay\cite{pgi,orderly}, who describes an effective method for computing the automorphism group of a graph and for generating a unique representative for each isomorphism class.  These methods were adapted by Leon in Ref.\cite{leon-main} to tackle similar problems, including several  group theory questions (which are graph isomorphism complete) as well as the problem at hand, namely computing the automorphism group of a linear code.  McKay and Leon have both written optimized software packages implementing these methods.  The first software package, which remains the standard program for determining graph isomorphism, is called \verb|nauty|, and is available online through McKay's website\cite{bdm-url}.    The second is now licensed under the GPL, thanks to the efforts of David Joyner and Vera Pless, and is available as a part of GUAVA\cite{guava}, a GAP package for coding theory. GPL implementations for the cases of graphs of any size and binary codes of length up to 32 (or 64 on many machines) can be found in Sage\cite{sage} ({\small\verb|sage.graphs.graph_isom|} and {\small\verb|sage.coding.binary_code|}).  For more details about the latter, written specifically for the classification of Adinkra topologies, the Reader can consult Section \ref{new_soft}.  The second and fourth sections are essentially a summary of the techniques described in Refs.\cite{pgi} and\cite{orderly}, respectively, which provide the full story.  The Reader should also be aware of the improvements to this algorithm in the sparse case in Ref.\cite{saucy}.

\subsection{The Mass Formula}
 The following formula, due to Gaborit\cite{rPGMass}, gives the number $\sigma(N,k)$ of 
distinct doubly even binary $[N,k]$ codes, depending on the residue of $N$ modulo 8:
\begin{align}
  \sigma(N,k) &= 
 \begin{cases}
  \ddd\prod_{i=0}^{k-1}\frac{2^{N-2i-2}+2^{\left\lfloor\frac{N}{2}\right\rfloor-i-1}-1}{2^{i+1} - 1},
      &\quad\text{if $N \equiv 1,7 \pmod 8,$} \\[6mm]
  \ddd\prod_{i=0}^{k-1}\frac{\bigl(2^{\frac{N}{2}-i-1}+1\bigr)\bigl(2^{\frac{N}{2}-i-1}-1\bigr)}{2^{i+1} - 1},
      &\quad\text{if $N \equiv 2,6 \pmod 8,$} \\[6mm]
  \ddd\prod_{i=0}^{k-1}\frac{2^{N-2i-2}-2^{\left\lfloor\frac{N}{2}\right\rfloor-i-1}-1}{2^{i+1} - 1},
      &\quad\text{if $N \equiv 3,5 \pmod 8,$} \\[6mm]
  \ddd\prod_{i=0}^{k-2}\frac{2^{N-2i-2}+2^{\frac{N}{2}-i-1}-2}{2^{i+1} - 1}
     \cdot\vp^+(N,k),
      &\quad\text{if $N \equiv 0 \pmod 8,$} \\[6mm]
 \ddd\prod_{i=0}^{k-2}\frac{2^{N-2i-2}+2^{\frac{N}{2}-i-1}-2}{2^{i+1} - 1}
     \cdot\vp^-(N,k),
      &\quad\text{if $N \equiv 4 \pmod 8,$}
 \end{cases}\\[4mm]
 \vp^\pm(N,k)
     &\Defl\frac{1}{2^{k-1}} + \frac{2^{N-2k} \pm 2^{\frac{N}{2} - k} - 2}{2^k - 1}~,
\end{align}
We are interested in those codes with $N\leq 32$. According to the mass formula above, we find the most such codes when $N=32$ and $k=10$.  In this case,
\begin{equation}
 \begin{aligned}
  \sigma(32, 10) &= 162,953,548,221,364,911,292,708,847,668,107,902,745,573,601,875 \\
   &\approx 1.6 \times 10^{47}.
 \end{aligned}
\end{equation}
However, we are interested not in the codes themselves, but rather their equivalence classes under permutations of the columns.
To establish a lower bound on the number of equivalence classes, suppose that every such code has a trivial automorphism group and thus its orbit is as large as possible, namely $32!$. Hence there are at least 
\begin{equation}
\bigg\lceil\frac{\sigma(32, 10)}{32!}\bigg\rceil = 619,287,158,132
\end{equation}
distinct classes of $[32, 10]$ codes.  If we carry out a similar computation for each case of interest, we find that there are at least
\begin{equation}
\sum_{k=0}^{\vk(32)}\bigg\lceil \frac{\sigma(32,k)}{32!}\bigg\rceil=1,117,005,776,858
~\approx~1.1\times10^{12}
\end{equation}
distinct classes (note that for smaller $N$, we simply add zero columns to obtain an equivalent code of degree 32).  Since the automorphism group of a generic code is small, this lower bound is a reasonable estimate for the actual number of codes.

 Regardless of the speed at which an individual processor can generate codes, the procedure for generating the codes can be run in many parallel processes.  We can not only divide by the rate at which we can generate codes per processor, but we can also divide by the number of processors working on the job.  This is what makes the computation feasible; speed increases almost linearly as a function of the number of processors.

\subsection{Computing the Automorphism Group}
 Computing the automorphism group of a binary code is closely related to the graph isomorphism problem.  Given a binary block code $C \subset \{0,1\}^N,$ define a bipartite graph $G(C)$, with the vertices partitioned into ``left'' and ``right'', as follows:  The words of the code are the left vertices, and the set $\{1, 2, ..., N\}$ forms the right set of vertices.  Given a word $w$ on the left, and a number $j$ on the right, there is an edge between the two if and only if $w$ has a 1 in the $j^\text{th}$ place.  Consider the automorphism group $\Aut(G(C))$ of the graph, but allow only those permutations that map left vertices to left vertices and right vertices to right vertices.  This gives a subgroup $\Aut(G(C))_b$ which is isomorphic in a canonical way to the permutation automorphism group of the code, simply by considering each permutation's action on the right set of vertices, which is identified with the set of columns of the code.  Our approach is to think of the codes in terms of their corresponding bipartite graphs, with additional structure coming from the linear, self-orthogonal code. The algorithm described in this section computes the automorphism group of the code, as well as a canonical representative of the code, which is an arbitrary but fixed representative of the isomorphism class.
 
 Suppose $G$ is a graph, and $\Pi$ is a partition of the vertices $V(G)$.  The partition $\Pi$ is called an \em equitable partition \em if for every pair of cells $C_1, C_2$ in $\Pi$, the number of edges $\{u, v\}$ such that $v \in C_2$ is constant as $u$ ranges over vertices in $C_1$.  By considering the orbits of vertices of the graph under any subgroup of the automorphism group, one obtains a partition which is always equitable.  However, the converse is not always true; there are equitable partitions which do not arise in this way.  One partition $\Pi_1$ is \em coarser \em than another partition $\Pi_2$ (or $\Pi_2$ is \em finer \em than $\Pi_1$) if every cell of $\Pi_2$ is a subset of a cell of $\Pi_1$.  The \em discrete partition \em is the one where every cell is of size one, and the \em unit partition \em is the one where there is only one cell.
 
 Given a graph $G$ and a partition $\Pi$ of $V(G)$, denote by $E_G(\Pi)$ the coarsest equitable partition of $G$ which is finer than $\Pi$.  In particular, if $\Pi$ is equitable with respect to $G$, then $E_G(\Pi) = \Pi$.  The algorithm described in this section takes as input a graph $G$ and a partition $\Pi_0$ of $V(G)$, and it returns the subgroup $\Aut(G)_{\Pi_0}$ of the automorphism group consisting of permutations that respect $\Pi_0$, \ie, that do not carry any vertex of one cell of $\Pi_0$ into another.  From here on, consider all partitions to be ordered, so that for example $(\{1,2,3\},\{4,5\}) \neq (\{4,5\},\{1,2,3\})$.  In particular, a discrete ordered partition is simply an ordering on the vertices of $G$.  For $v \in V(G)$, if $\Pi = (C_1, ..., C_r)$ and $v \in C_i$, define $R(G, \Pi, v) = E_G(\Pi^\prime)$, where $\Pi^\prime = (C_1, ..., C_{i-1}, \{v\}, C_{i-1}\setminus\{v\}, C_{i+1}, ..., C_r)$.  This defines $R(G, \Pi, v)$ up to reordering of the cells of the partition. For full details, see Ref.\cite{pgi}, Algorithm 2.5, which defines the ordering.
 
 Define a rooted tree $T = T(G, \Pi_0)$ consisting of equitable partitions of $G$ finer than $\Pi_0$ as follows: The root of $T$ is the partition $E_G(\Pi_0)$, and for any node $\Pi$ of $T$, its children are the partitions $R(G, \Pi, v)$ for which $v$ is not yet in a singleton cell of $\Pi$.  The group $\Aut(G)_{\Pi_0}$ acts on the tree $T$ by taking a sequence of nested partitions $\Pi_0, ..., \Pi_k$ to the resulting nested sequence of partitions by letting $\Aut(G)_{\Pi_0}$ act on the elements of the cells; since this sequence is defined by a sequence $v_1, ..., v_k$, the sequence $\gamma(v_1), ..., \gamma(v_k)$ defines the image under $\gamma \in \Aut(G)_{\Pi_0}$.  Further, the subgroup of permutations respecting a partition $\Pi$ acts on any subtree of $T$ rooted at $\Pi$.  Because this action is faithful, the structure of $T$ can be used to calculate a set of generators for $\Aut(G)_{\Pi_0}$.
 
 The algorithm itself is a backtrack algorithm that successively refines the partitions to explore the tree $T$.  Any leaf of the tree is a discrete ordered partition $(\{v_{1}\}, \{v_{2}\}, ..., \{v_{n}\})$, which defines an ordering of the vertices of $G$.  Given two leaves of the tree $T$, one has two orderings $v_1, ..., v_n$ and $v_1^\prime, ..., v_n^\prime$, which we think of as a permutation defined by $v_i \mapsto v_i^\prime$.  This is the means by which the algorithm finds automorphisms, and it uses the presence of automorphisms to deduce when different parts of the tree are equivalent.  At this point the algorithm backtracks towards the root until there is a new part of the tree to explore which is not yet known to be equivalent to a part of the tree already traversed.  In practice, the part of the tree traversed is much smaller than the entire tree, and once one backtracks off of the root, one has a set of generators for $\Aut(G)_{\Pi_0}$.  Invariants and orderings can also be used to find a leaf of the tree $T$ which is maximal amongst the nodes with largest invariants, which is uniquely defined independently of reordering the inputs.  This leaf is a ``canonical label'' for the pair $(G, \Pi_0)$.  In particular, if $\gamma \in \Aut(G)_{\Pi_0}$, then the canonical label for $(G^\gamma, \Pi_0^\gamma)$ is the same as that of $(G, \Pi_0)$.  For details, the Reader is once again directed to the theory in Ref.\cite{pgi} and the code in Ref.\cite{sage}.
 
 The outcome is an algorithm to efficiently compute the group $\Aut(C) = \Aut(G(C))_b$ (here ``='' means canonically isomorphic) as well as a unique representative of each isomorphism class which we will denote $c(C)$, both of which will be used in the algorithm to generate all the permutation classes of doubly even codes.  Also note that by the use of the bipartite graph construction, not only is $c(C)$ a binary code, but it also contains an ordering of its words.  In practice, $c(C)$ will be output as an actual generator matrix.

\subsection{New Software for Special Circumstances}
\label{new_soft}
 It is a curious coincidence that the limit on $N$ for this application is 32.  This implies that the size of the words in a code are at most the size of the machine words on a 32-bit system.  The analogy goes further: the operation of adding two vectors in a code becomes the single clock tick of taking an \verb|XOR| on the machine words representing them.

 When analyzing codes in terms of their corresponding graphs, we can take advantage of the linear orthogonal structure of the codes to optimize standard graph algorithms.
For example, given an equitable partition $\Pi$ of the bipartite graph $G(C)$, the set of words appearing in their own singleton cells of $\Pi$ is closed under binary addition.  This is immediate from the definition of equitable partition and of binary addition, and this fact greatly reduces the expected height of the tree to be searched in computing $c(C)$.  Further, every permutation acts linearly, so to check if some permutation is indeed an automorphism, we need only check it on a basis.  It is also possible to use the linear structure to derive variants on the refinement procedure, depending on the size and type of binary codes being generated.
 
 In the generation algorithm, since we are interested only in self-orthogonal codes, we can use more linear algebra to efficiently enumerate the possible children of a code. We start by writing the code in standard form, $C = [I_k|C^\prime]$, and then extend the basis given by the rows to a basis for $C^\perp$.  This basis can be taken so that the first $k$ positions of each additional vector are zero, and we can also shuffle the pivot columns to the front, so that we obtain this basis as the rows of a matrix of the form $$\left[\begin{array}{c|c}I_k & C^{\prime\prime} \\\hline0 & \begin{array}{c|c}I_{N-2k} & *\end{array}\end{array}\right],$$ where $C^{\prime\prime}$ is a rearrangement of the columns of $C^\prime$.  The rowspan of the bottom part of the matrix, $\left[\begin{array}{ccc}0 & I_{N-2k} & *\end{array}\right]$, then forms a set of unique coset representatives of the original code.  Furthermore, (\ref{eqn:inclusionexclusion}) implies the following:
\begin{corollary} Suppose $C$ is a binary code spanned by $\{v_1,\cdots,v_k\}$, and that $\wt(v_i) \equiv 0 \pmod 4$ for each $i = 1,\dots,k$. If either of the following hold for all $i \neq j$, then $C$ is doubly even:
\begin{itemize}
\item[a\em)] $\langle v_i, v_j \rangle = 0$,
\item[b\em)] $\wt(v_i + v_j) \equiv 0 \pmod 4$.
\end{itemize}
\end{corollary}
\begin{proof} These results are special to binary codes, since every vector in $C$ is the sum of distinct vectors in $\{v_1,\cdots,v_k\}$. Equation (\ref{eqn:inclusionexclusion}) implies that both conditions are equivalent, so we use {\em b}). Suppose $x,y,z \in C$ are such that
\[\wt(x) \equiv \wt(y) \equiv \wt(z) \equiv \wt(x+y) \equiv \wt(x+z) \equiv \wt(y+z) \equiv 0 \pmod 4.\]
Then again by (\ref{eqn:inclusionexclusion}),
\begin{align*}
\wt(x + y + z) \equiv&\ \wt(x + y) + \wt(z) - 2\langle x+y, z \rangle\\
\equiv& - 2\langle x, z \rangle - 2\langle y, z \rangle\\
\equiv& - 2\wt(x) - 2\wt(z) + 2\wt(x+z)\\
& - 2\wt(y) - 2\wt(z) + 2\wt(y+z)\\
\equiv&\ 0 \pmod 4.
\end{align*}
By induction on the number of basis elements in a linear combination, $C$ is doubly even.
\end{proof}
Going back to the situation at hand, this means that we need only examine the doubly even vectors in the rowspan of $\left[\begin{array}{ccc}0 & I_{N-2k} & *\end{array}\right]$, since the code we already have is doubly even and self-orthogonal, and any vector to be considered is already orthogonal to the code.  These vectors are in one-to-one correspondence with the set of doubly even codes of one dimension higher containing (our rearrangement of) $C$ .  We can do even better by starting with the even subcode of $\left[\begin{array}{ccc}0 & I_{N-2k} & *\end{array}\right]$, since the sum of two even words is even.

\subsection{Exhaustively Generating the Codes}
 Here we use another algorithm developed by McKay\cite{orderly}, called canonical augmentation.  The ingredients for this are a question like the one we are considering, a canonical labeling function as described, a way of computing the automorphism group of a code, and a hereditary structure on the objects to be generated.
 
 We obtain a hereditary structure by thinking of a certain code's \em children \em as the set of codes constructible by adding a single word not already in the code. Thus the dimension of a code is one less than the dimension of all its children, and every code can be built up from the zero dimensional code by a limited number of augmentations.  In our application, this puts all the desired codes on a tree of height 16.  At this point one can consider the very naive algorithm of generating all the children for each code, keeping a list and throwing out isomorphs.  However, the obvious problem with this approach is the excessive isomorphism computation, which is expensive.  This problem is not really ameliorated by storing the canonical representative of each code, as the canonical representative calculation is almost as expensive as isomorphism testing.

 The idea behind canonical augmentation is that instead of requiring the objects generated be in canonical form, we require simply that the augmentation itself be canonical.  The definition of an \em augmentation \em is simply an ordered pair $(C, C^\prime)$, where $C^\prime$ is a child of $C$.  An isomorphism of augmentations is a permutation $\gamma$ such that $(C^\gamma,{C^\prime}^\gamma) = (D,D^\prime)$.  For example, $(C, C^\prime) \cong (C, C^{\prime\prime})$ implies that there is a $\gamma \in \Aut(C)$ such that ${C^\prime}^\gamma = C^{\prime\prime}$.
 
 Suppose we have a function $p$ which takes codes $C$ of dimension $k > 0$ to codes $p(C)$ of dimension $k-1$, such that $C$ is a child of $p(C)$.  Suppose further that $p$ satisfies the following property: if $C \cong D$, then $(p(C), C) \cong (p(D), D)$.  If we have such a function, we can define that an augmentation $(C, D)$ is canonical if $(C, D) \cong (p(D), D)$.  Now suppose that we have two augmentations $(C, C^\prime)$ and $(D, D^\prime)$ such that $C^\prime \cong D^\prime$.  If they were both canonical augmentations then we would have
 $$(C, C^\prime) \cong \bigl(p(C^\prime), C^\prime\bigr) \cong \bigl(p(D^\prime), D^\prime\bigr) \cong (D, D^\prime).$$
  In other words, if both augmentations are canonical, then $C^\prime \cong D^\prime$ implies that $C \cong D$.  Thus any repeated isomorphs would have to be due to the parents being repeated.  If we assume that there are no isomorphs on the parent level, then this implies that $C = D$, and $(C, C^\prime) \cong (C, D^\prime)$.  In other words, under the framework of canonical augmentation, repeated isomorphs arise only due to automorphisms of the parents.  Thus we have Algorithm \ref{alg_gen} below.
 \begin{algorithm}
 \caption{Generate all doubly even binary codes of degree 32, isomorph-free.}
 \label{alg_gen}
 \begin{verbatim}

 C_0 := the dimesion zero code, of degree 32
 traverse(C_0)
 procedure traverse(C):
     report C
     children := the children of C
     representatives := {}
     for D in children:
         if D is minimal in its orbit under Aut(C):
             add D to representatives
     for D in representatives:
         if (C, D) is a canonical augmentation:
             traverse(D)
\end{verbatim}\end{algorithm}

The only remaining question is to produce such a function $p$.  This is where we use the canonical label defined in the last section.  If $C$ is a code of dimension $k$, let $\gamma$ be the permutation taking $C$ to $c(C)$.  Recall that $c(C)$ came with a particular generator matrix determined by the ordering of the words of the code.  Removing the last row of that generator matrix gives a code $C^\prime$ of the same dimension as $C$ and applying $\gamma^{-1}$ to that, we arrive at $p(C)$.  Suppose that $C \cong D$, which in particular implies that $c(C) = c(D)$.  Let $\gamma_C, \gamma_D$ be the permutations taking $C,D$ to $c(C), c(D)$, respectively.  These are isomorphisms, so in fact, $\gamma_D^{-1}\circ\gamma_C$ is an isomorphism from $C$ to $D$ taking $p(C)$ to $p(D)$ by the definition of $p$.  This proves the property that $(p(C),C) \cong (p(D),D)$.  Further, since we have already done the computation $c(C)$, we can obtain not only the canonical label but also generators for the automorphism group $\Aut(C)$ for free.  Thus when we are checking whether $(C^\prime, C) \cong (p(C), C)$, we can simply look for an element of $\Aut(C)$ that takes $C^\prime$ to $p(C)$.

 The final task of enumerating the codes for storage to disk will approximate the experience of a harvest.  Many separate worker processes will be working in parallel, each examining the pairs $(C, D)$ as in the second to last line of Algorithm \ref{alg_gen}, for a fixed $C$ and many $D$.  Each worker will receive one code at a time, and perform all the steps in Algorithm \ref{alg_gen}, only instead of recursively calling the function on the last line, they will record the augmented codes that succeed in a table, which will each in turn become the fodder for another worker.  As codes are given to workers to augment on, they will be flagged as searched, and once all the codes have been flagged and all the workers are done, all the desired codes will be in the table.  The intended end result of all this is a searchable online database which categorizes the codes by $N$ and $k$,  automorphism group size, weight distribution, and perhaps other parameters.

\vfill\clearpage
\bibliographystyle{elsart-numX}
\bibliography{Refs}

\def\rasp{\leavevmode\raise.45ex\hbox{$\rhook$}}
\begin{thebibliography}{10}
 \expandafter\ifx\csname url\endcsname\relax
 \raggedright
 \def\url#1{{\small\texttt{#1}}}\fi
 \expandafter\ifx\csname urlprefix\endcsname\relax\def\urlprefix{URL: }\fi

\bibitem{rGR0}
S.~J. Gates, Jr., L.~Rana, {\em Ultramultiplets: A new representation of rigid
  2-d, ${N}$=8 supersymmetry}, Phys. Lett. B342 (1995) 132--137.
\url{hep-th/9410150}

\bibitem{rMMYau1}
S.-T. Yau (Ed.), {\em Mirror Manifolds}, International Press, 1990.

\bibitem{rMMYau2}
B.~Greene, S.-T. Yau (Eds.), {\em Mirror Manifolds {II}}, International Press,
  1996.

\bibitem{rMMYau3}
D.~H. Phong, L.~Vinet, S.-T. Yau (Eds.), {\em Mirror Manifolds {III}}, American
  Mathematical Society, Providence, RI, 1999.

\bibitem{rMS}
K.~Hori, S.~Katz, A.~Klemm, R.~Pandharipande, R.~Thomas, C.~Vafa, R.~Vakil,
  E.~Zaslow, {\em Mirror symmetry}, Vol.~1 of {\em Clay Mathematics
  Monographs}, American Mathematical Society, Providence, RI, 2003.

\bibitem{rDon}
S.~K. Donaldson, {\em An application of gauge theory to the topology of
  four-manifolds}, Journal of Differential Geometry 18 (1983) 269--316.

\bibitem{rTQFT}
E.~Witten, {\em Topological quantum field theory}, Communications in
  Mathematical Physics 117 (1988) 353--386.

\bibitem{rSW}
E.~Witten, {\em Monopoles and four-manifolds}, Mathematical Research Letters 1
  (1994) 769--796.

\bibitem{rKISurv}
K.~Iga, {\em What do topologists want from {S}eiberg-{W}itten theory?},
  Internat. J. Modern Phys. A 17~(30) (2002) 4463--4514.

\bibitem{rGR1}
S.~J. Gates, Jr., L.~Rana, {\em A theory of spinning particles for large
  {$N$}-extended supersymmetry}, Phys. Lett. B 352~(1-2) (1995) 50--58.
\url{hep-th/9504025}

\bibitem{rGR2}
S.~J. Gates, Jr., L.~Rana, {\em A theory of spinning particles for large
  {$N$}-extended supersymmetry. {II}}, Phys. Lett. B 369~(3-4) (1996) 262--268.
\url{hep-th/9510151}

\bibitem{rGLP}
S.~J. Gates, Jr., W.~D. Linch, III, J.~Phillips, {\em When superspace is not
  enough}.
\url{hep-th/0211034}

\bibitem{rA}
M.~Faux, S.~J. Gates, Jr., {\em Adinkras: A graphical technology for
  supersymmetric representation theory}, Phys. Rev. D (3) 71 (2005) 065002.
\url{hep-th/0408004}

\bibitem{rFKS}
M.~Faux, D.~Kagan, D.~Spector, {\em Central charges and extra dimensions in
  supersymmetric quantum mechanics}.
\url{hep-th/0406152}

\bibitem{rFS1}
M.~Faux, D.~Spector, {\em Duality and central charges in supersymmetric quantum
  mechanics}, Phys. Rev. D (3) 70~(8) (2004) 085014, 5.
\url{hep-th/0311095}

\bibitem{rFS2}
M.~Faux, D.~Spector, {\em A bps interpretation of shape invariance}, J. Phys.
  A37 (2004) 10397--10407.
\url{quant-ph/0401163}

\bibitem{rKRT}
Z.~Kuznetsova, M.~Rojas, F.~Toppan, {\em Classification of irreps and
  invariants of the ${N}$-extended supersymmetric quantum mechanics}, JHEP 03
  (2006) 098.
\url{hep-th/0511274}

\bibitem{rBKMO}
S.~Bellucci, S.~Krivonos, A.~Marrani, E.~Orazi, {\em `{R}oot' action for {N}=4
  supersymmetric mechanics theories}, Phys. Rev. D 73 (2006) 025011.
\url{hep-th/0511249}

\bibitem{rBG}
S.~Bellucci, S.~J. Gates, Jr., E.~Orazi, {\em A journey through garden
  algebras}, Lect. Notes Phys. 698 (2006) 1--47.
\url{hep-th/0602259}

\bibitem{rBKLS}
S.~Bellucci, S.~Krivonos, O.~Lechtenfeld, A.~Shcherbakov, {\em Superfield
  formulation of nonlinear {N}=4 supermultiplets}.
\url{arXiv:0710.3832 [hep-th]}

\bibitem{r6-1}
C.~F. Doran, M.~G. Faux, S.~J. Gates, Jr., T.~H{\"u}bsch, K.~M. Iga, G.~D.
  Landweber, {\em On graph-theoretic identifications of {A}dinkras,
  supersymmetry representations and superfields}, Int. J. Mod. Phys. A22 (2007)
  869--930.
\url{math-ph/0512016}

\bibitem{r6-2}
C.~F. Doran, M.~G. Faux, S.~J. Gates, Jr., T.~H{\"u}bsch, K.~M. Iga, G.~D.
  Landweber, {\em Adinkras and the dynamics of superspace prepotentials}, Adv.
  S. Th. Phys. 2~(3) (2008) 113--164.
\url{hep-th/0605269}

\bibitem{r6-3a}
C.~F. Doran, M.~G. Faux, S.~J. Gates, Jr., T.~H{\"u}bsch, K.~M. Iga, G.~D.
  Landweber, {\em A counter-example to a putative classification of
  1-dimensional, ${N}$-extended supermultiplets}, Adv. S. Th. Phys. 2~(3)
  (2008) 99--111.
\url{hep-th/0611060}

\bibitem{r6-4}
C.~F. Doran, M.~G. Faux, S.~J. Gates, Jr., T.~H{\"u}bsch, K.~M. Iga, G.~D.
  Landweber, {\em On the matter of ${N}=2$ matter}, Phys. Lett. B 659 (2008)
  441--446.
\url{arXiv/0710.5245}

\bibitem{rMcWS}
F.~J. MacWilliams, N.~J.~A. Sloane, {\em The Theory of Error-Correcting Codes},
  North-Holland, Amsterdam, 1977.

\bibitem{rCHVP}
W.~C. Huffman, V.~Pless, {\em Fundamentals of Error-Correcting Codes},
  Cambridge Univ. Press, 2003.

\bibitem{rCPS}
J.~H. Conway, V.~Pless, N.~J.~A. Sloane, {\em The binary self-dual codes of
  length up to 32: A revised enumeration}, J. Combinatorial Theory, Series A 60
  (1992) 183--195.

\bibitem{rHSS}
T.~H{\"u}bsch, {\em Haploid {$(2,2)$}-superfields in {$2$}-dimensional
  space-time}, Nucl. Phys. B555~(3) (1999) 567--628.

\bibitem{rGSS}
R.~Q. Almukahhal, T.~H{\"u}bsch, {\em Gauging {Y}ang-{M}ills symmetries in
  {$(1{+}1)$}-dimensional space-time}, Internat. J. Modern Phys. A 16~(29)
  (2001) 4713--4768.

\bibitem{r6-7a}
C.~F. Doran, M.~G. Faux, S.~J. Gates, Jr., T.~H{\"u}bsch, K.~M. Iga, G.~D.
  Landweber, {\em Super-zeeman embedding models on {N}-supersymmetric
  world-lines}.
\url{arXiv/0803.3434}

\bibitem{r6--1}
C.~F. Doran, M.~G. Faux, S.~J. Gates, Jr., T.~H{\"u}bsch, K.~M. Iga, G.~D.
  Landweber, {\em Off-shell supersymmetry and filtered {C}lifford
  supermodules}.
\url{math-ph/0603012}

\bibitem{r6-3c}
C.~F. Doran, M.~G. Faux, S.~J. Gates, Jr., T.~H{\"u}bsch, K.~M. Iga, G.~D.
  Landweber, {\em Relating doubly-even error-correcting codes, graphs, and
  irreducible representations of ${N}$-extended supersymmetry}, {\em in:
  F.~Liu, et~al. (Eds.), Discrete and Computational Mathematics}, Nova Science
  Publishers, Inc., Hauppauge, NY, 2008.

\bibitem{r6-3.2}
C.~F. Doran, M.~G. Faux, S.~J. Gates, Jr., T.~H{\"u}bsch, K.~M. Iga, G.~D.
  Landweber, R.~L. Miller, {\em Adinkras for clifford algebras, and some
  one-dimensional supermultiplets}.
\url{forthcoming}

\bibitem{rNEConstrA}
N.~D. Elkies, {\em Lattices, linear codes, and invariants, {P}art {II}},
  Notices of the {AMS} 47~(11) (2000) 1382--1391.

\bibitem{rBilRees}
R.~T. Bilous, G.~H.~J. van Rees, {\em An enumeration of binary self-dual codes
  of length 32}, Designs, Codes and Cryptography 26~(1-3) (2002) 66--86.

\bibitem{rGHR}
S.~J. Gates, Jr., C.~M. Hull, M.~Ro\v{c}ek, {\em Twisted multiplets and new
  supersymmetric nonlinear sigma models and new supersymmetric nonlinear sigma
  models}, Nucl. Phys. B248 (1984) 157.

\bibitem{rPGMass}
P.~Gaborit, {\em Mass formulas for self-dual codes over ${Z}_4$ and
  ${F}_q+u{F}_q$ rings}, IEEE Trans. Inform. Theory 42~(4) (1996) 1222--1228.

\bibitem{rGLPR}
S.~J. Gates, Jr., W.~Linch, J.~Phillips, L.~Rana, {\em The fundamental
  supersymmetry challenge remains}, Gravit. Cosmol. 8~(1-2) (2002) 96--100.
\url{hep-th/0109109}

\bibitem{rLM}
H.~B. Lawson, Jr., M.-L. Michelsohn, {\em Spin geometry}, Vol.~38 of {\em
  Princeton Mathematical Series}, Princeton University Press, Princeton, NJ,
  1989.

\bibitem{rKT07}
Z.~Kuznetsova, F.~Toppan, {\em Refining the classification of the irreps of the
  1d n- extended supersymmetry}.
\url{hep-th/0701225}

\bibitem{pgi}
B.~McKay, {\em Practical graph isomorphism}, Congr. Numer. 30~(30) (1981)
  45--87.

\bibitem{orderly}
B.~McKay, {\em Isomorph-free exhaustive generation}, Journal of Algorithms 26
  (1998) 306--324.

\bibitem{leon-main}
J.~S. Leon, {\em Permutation group algorithms based on partition, {I}: Theory
  and algorithm}, Journal of Symbolic Computation 12 (1991) 533--583.

\bibitem{bdm-url}
B.~McKay, {\em The nauty page}, \url{http://cs.anu.edu.au/~bdm/nauty/} (2008).

\bibitem{guava}
R.~Baart, T.~Boothby, J.~Cramwinckel, W.~Joyner, R.~Miller, E.~Minkes,
  E.~Roijackers, L.~Ruscio, C.~Tjhai, {\em GAP package GUAVA},
  \url{http://opensourcemath.org/guava/} (2008).

\bibitem{sage}
W.~Stein, {\em {Sage}: {O}pen {S}ource {M}athematical {S}oftware ({V}ersion
  3.0.2)}, The Sage~Group, \url{http://www.sagemath.org} (2008).

\bibitem{saucy}
P.~Darga, M.~Liffiton, K.~Sakallah, I.~Markov, {\em Exploiting structure in
  symmetry detection for {CNF}}, Design Automation Conference, 2004.
  Proceedings. 41st (2004) 530--534.
\url{http://vlsicad.eecs.umich.edu/BK/SAUCY/}

\end{thebibliography}
\end{document}